\documentclass[11pt]{article}
\usepackage[margin=1in]{geometry}
\usepackage{graphicx}
\usepackage[pdftex,colorlinks=true,linkcolor=blue,citecolor=blue,urlcolor=black]{hyperref}
\usepackage{amsmath, amsthm, amssymb}
\usepackage{subfigure}
\usepackage{comment}
\usepackage{url}
\usepackage{pdflscape}
\usepackage[ruled,lined,linesnumbered]{algorithm2e}
\usepackage{tikz}
\usepackage{calc}

\newcommand{\R}{\mathbb{R}}

\newcommand{\Z}{\mathbb{Z}}

\newcommand{\E}{\mathbb{E}}

\newcommand{\ket}[1]{| #1 \rangle}

\newcommand{\ip}[2]{\langle #1|#2 \rangle}

\newcommand{\proj}[1]{| #1 \rangle \langle #1 |}
\newcommand{\bracket}[3]{\langle #1|#2|#3 \rangle}

\DeclareMathOperator{\tr}{tr}

\DeclareMathOperator{\Var}{Var}

\newcommand{\be}{\begin{equation}}
\newcommand{\ee}{\end{equation}}
\newcommand{\bea}{\begin{eqnarray}}
\newcommand{\eea}{\end{eqnarray}}
\newcommand{\bes}{\begin{equation*}}
\newcommand{\ees}{\end{equation*}}
\newcommand{\beas}{\begin{eqnarray*}}
\newcommand{\eeas}{\end{eqnarray*}}

\newtheorem{thm}{Theorem}
\newtheorem*{thm*}{Theorem}
\newtheorem{cor}[thm]{Corollary}

\newtheorem{lem}[thm]{Lemma}
\newtheorem*{lem*}{Lemma}
\newtheorem{prop}[thm]{Proposition}

\newtheorem{claim}[thm]{Claim}
\newtheorem{dfn}{Definition}

\newcommand{\boxalgm}[3]{
\renewcommand{\figurename}{Algorithm}
\begin{figure}[htp]
\noindent \framebox{
\begin{minipage}{\textwidth-0.5cm}
#3
\end{minipage}
}
\caption{#2}
\label{#1}
\end{figure}
\renewcommand{\figurename}{Figure}
}

\newcommand{\boxalgmfig}[3]{
\begin{figure}
\noindent \framebox{
\begin{minipage}{\textwidth-0.5cm}
#3
\end{minipage}
}
\caption{#2}
\label{#1}
\end{figure}
}

\begin{document}

%\begin{bibunit}[plain]

\title{Quantum speedup of Monte Carlo methods}
\author{Ashley Montanaro\thanks{Department of Computer Science, University of Bristol, UK; {\tt ashley.montanaro@bristol.ac.uk}.}}
\maketitle

\begin{abstract}
Monte Carlo methods use random sampling to estimate numerical quantities which are hard to compute deterministically. One important example is the use in statistical physics of rapidly mixing Markov chains to approximately compute partition functions. In this work we describe a quantum algorithm which can accelerate Monte Carlo methods in a very general setting. The algorithm estimates the expected output value of an arbitrary randomised or quantum subroutine with bounded variance, achieving a near-quadratic speedup over the best possible classical algorithm. Combining the algorithm with the use of quantum walks gives a quantum speedup of the fastest known classical algorithms with rigorous performance bounds for computing partition functions, which use multiple-stage Markov chain Monte Carlo techniques. The quantum algorithm can also be used to estimate the total variation distance between probability distributions efficiently.
\end{abstract}

%\maketitle
% ------------------------------------------------------------------------------

\section{Introduction}

Monte Carlo methods are now ubiquitous throughout science, in fields as diverse as statistical physics~\cite{krauth06}, microelectronics~\cite{jacoboni89} and mathematical finance~\cite{glasserman03}. These methods use randomness to estimate numerical properties of systems which are too large or complicated to analyse deterministically. In general, the basic core of Monte Carlo methods involves estimating the expected output value $\mu$ of a randomised algorithm $\mathcal{A}$. The natural algorithm for doing so is to produce $k$ samples, each corresponding to the output of an independent execution of $\mathcal{A}$, and then to output the average $\widetilde{\mu}$ of the samples as an approximation of $\mu$. Assuming that the variance of the random variable corresponding to the output of $\mathcal{A}$ is at most $\sigma^2$, the probability that the value output by this estimator is far from the truth can be bounded using Chebyshev's inequality:
\[ \Pr[|\widetilde{\mu} - \mu| \ge \epsilon] \le \frac{\sigma^2}{k \epsilon^2}. \]
It is therefore sufficient to take $k = O(\sigma^2 / \epsilon^2)$ to estimate $\mu$ up to additive error $\epsilon$ with, say, 99\% success probability. This simple result is a key component in many more complex randomised approximation schemes (see e.g.~\cite{stefankovic09,krauth06}).

Although this algorithm is fairly efficient, its quadratic dependence on $\sigma / \epsilon$ seems far from ideal: for example, if $\sigma = 1$, to estimate $\mu$ up to 4 decimal places we would need to run $\mathcal{A}$ over 100 million times. Unfortunately, it can be shown that, without any further information about $\mathcal{A}$, the sample complexity of this algorithm is asymptotically optimal~\cite{dagum00} with respect to its scaling with $\sigma$ and $\epsilon$, although it can be improved by a constant factor~\cite{huber14}.

We show here that, using a quantum computer, the number of uses of $\mathcal{A}$ required to approximate $\mu$ can be reduced almost quadratically beyond the above classical bound. Assuming that the variance of the output of the algorithm $\mathcal{A}$ is at most $\sigma^2$, we present a quantum algorithm which estimates $\mu$ up to additive error $\epsilon$, with $99\%$ success probability, using $\mathcal{A}$ only $\widetilde{O}(\sigma/ \epsilon)$ times\footnote{The $\widetilde{O}$ notation hides polylogarithmic factors.}. It follows from known lower bounds on the quantum complexity of approximating the mean~\cite{nayak99} that the runtime of this algorithm is optimal, up to polylogarithmic factors. This result holds for an {\em arbitrary} algorithm $\mathcal{A}$ used as a black box, given only an upper bound on the variance.

An important aspect of this construction is that the underlying subroutine $\mathcal{A}$ need not be a classical randomised procedure, but can itself be a quantum algorithm. This enables any quantum speedup obtained by $\mathcal{A}$ to be utilised within the overall framework of the algorithm. A particular case in which this is useful is quantum speedup of Markov chain Monte Carlo methods~\cite{levin09}. Classically, such methods use a rapidly mixing Markov chain to approximately sample from a probability distribution corresponding to the stationary distribution of the chain. Quantum walks are the quantum analogue of random walks (see e.g.~\cite{venegas12} for a review). In some cases, quantum walks can reduce the mixing time quadratically (see e.g.~\cite{aharonov01,wocjan08}), although it is not known whether this can be achieved in general~\cite{richter07,aharonov07a,dunjko15}. We demonstrate that this known quadratic reduction can be combined with our algorithm to speed up the fastest known general-purpose classical algorithm with rigorous performance bounds~\cite{stefankovic09} for approximately computing partition functions up to small relative error, a fundamental problem in statistical physics~\cite{krauth06}. As another example of how our algorithm can be applied, we substantially improve the runtime of a quantum algorithm for estimating the total variation distance between two probability distributions~\cite{bravyi11a}.

% ------------------------------------------------------------------------------

\subsection{Prior work}
\label{sec:prior}

The topic of quantum estimation of mean output values of algorithms with bounded variance connects to several previously-explored directions. First, it generalises the problem of approximating the mean, with respect to the uniform distribution, of an arbitrary bounded function. This has been addressed by a number of authors. The first asymptotically optimal quantum algorithm for this problem, which uses $O(1/\epsilon)$ queries to achieve additive error $\epsilon$, seems to have been given by Heinrich~\cite{heinrich01}; an elegant alternative optimal algorithm was later presented by Brassard et al.~\cite{brassard11}. Previous algorithms, which are optimal up to lower-order terms, were described by Grover~\cite{grover98}, Aharonov~\cite{aharonov98} and Abrams and Williams~\cite{abrams99}. Using similar techniques to Brassard et al., Wocjan et al.~\cite{wocjan09} described an efficient algorithm for estimating the expected value of an arbitrary bounded observable. It is not difficult to combine these ideas to approximate the mean of arbitrary bounded functions with respect to nonuniform distributions (see Section \ref{sec:bounded}).

One of the main technical ingredients in the present paper is based on an algorithm of Heinrich for approximating the mean, with respect to the uniform distribution, of functions with bounded $L^2$ norm~\cite{heinrich01}. Section \ref{sec:boundedl2} describes a generalisation of this result to nonuniform distributions, using similar techniques. This is roughly analogous to the way that amplitude amplification~\cite{brassard02} generalises Grover's quantum search algorithm~\cite{grover97}.

The related problem of quantum estimation of expectation values of observables, an important task in the simulation of quantum systems, has been studied by Knill, Ortiz and Somma~\cite{knill07}. These authors give an algorithm for estimating $\tr(A \rho)$ for observables $A$ such that one can efficiently implement the operator $e^{-iAt}$. The algorithm is efficient (i.e.\ achieves runtimes close to $O(1/\epsilon)$) when the tails of the distribution $\tr(A \rho)$ decay quickly. However, in the case where one only knows an upper bound on the variance of this distribution, the algorithm does not achieve a better runtime than classical sampling. Yet another related problem, that of exact Monte Carlo {\em sampling} from a desired probability distribution, was addressed by Destainville, Georgeot and Giraud~\cite{destainville10}. Their quantum algorithm, which uses Grover's algorithm as a subroutine, achieves roughly a quadratic speedup over classical exact sampling. This algorithm's applicability is limited by the fact that its runtime scaling can be as slow as $O(\sqrt{N})$, where $N$ is the number of states of the system; we often think of $N$ as being exponential in the input size.

Quantum algorithms have been used previously to approximate classical partition functions and solve related problems. In particular, a number of authors~\cite{lidar97,lidar04,aharonov07,vandennest08a,geraci08,arad10,geraci10,delascuevas11,matsuo14} have considered the complexity of computing Ising and Potts model partition functions. These works in some cases achieve exponential quantum speedups over the best known classical algorithms. Unfortunately, they in general either produce an approximation accurate up to a specified {\em additive} error bound, or only work for specific classes of partition function problems with restrictions on interaction strengths and topologies, or both. Here we aim to approximate partition functions up to small relative error in a rather general setting.

Using related techniques to the present work, Somma et al.~\cite{somma08} used quantum walks to accelerate classical simulated annealing processes, and quantum estimation of partition functions up to small relative error was addressed by Wocjan et al.~\cite{wocjan09}. Their algorithm, which is based on the use of quantum walks and amplitude estimation, achieves a quadratic speedup over classical algorithms with respect to both mixing time and accuracy. However, it cannot be directly applied to accelerate the most efficient classical algorithms for approximating partition function problems, which use so-called Chebyshev cooling schedules (discussed in Section \ref{sec:partition}). This is essentially because these algorithms are based around estimating the mean of random variables given only a bound on the variance. This was highlighted as an open problem in~\cite{wocjan09}, which we resolve here.

Several recent works have developed quantum algorithms for the quantum generalisation of sampling from a Gibbs distribution: producing a Gibbs state $\rho \propto e^{-\beta H}$ for some quantum Hamiltonian $H$~\cite{tucci09,poulin09,temme11,yung12}. Given such a state, one can measure a suitable observable to compute some quantity of interest about $H$. Supplied with an upper bound on the variance of such an observable, the procedure detailed here can be used (as for any other quantum algorithm) to reduce the number of repetitions required to estimate the observable to a desired accuracy.

% ------------------------------------------------------------------------------

\subsection{Techniques}
\label{sec:approx}

\begin{table}
\begin{center}
\begin{tabular}{|c|l|l|c|}
\hline Algorithm & Precondition & Approximation of $\mu$ & Uses of $\mathcal{A}$ and $\mathcal{A}^{-1}$ \\
\hline \ref{alg:meanbounded} & $v(\mathcal{A}) \in [0,1]$ & Additive error $\epsilon$ & $O(1/\epsilon)$\\
\ref{alg:meanvariance} & $\Var(v(\mathcal{A})) \le \sigma^2 $ & Additive error $\epsilon$ & $\widetilde{O}(\sigma/\epsilon)$\\
\ref{alg:meanratio} & $\Var(v(\mathcal{A}))/(\E[v(\mathcal{A})])^2 \le B$ & Relative error $\epsilon$ & $\widetilde{O}(B/\epsilon)$\\
\hline
\end{tabular}
\end{center}
\caption{Summary of the main quantum algorithms presented in this paper for estimating the mean output value $\mu$ of an algorithm $\mathcal{A}$. (Algorithm \ref{alg:meansub}, omitted, is a subroutine used in Algorithm \ref{alg:meanvariance}.)}
\label{tab:algs}
\end{table}

We now give an informal description of our algorithms, which are summarised in Table \ref{tab:algs} (for technical details and proofs, see Section~\ref{sec:algs}). For any randomised or quantum algorithm $\mathcal{A}$, we write $v(\mathcal{A})$ for the random variable corresponding to the value computed by $\mathcal{A}$, with the expected value of $v(\mathcal{A})$ denoted $\E[v(\mathcal{A})]$. For concreteness, we think of $\mathcal{A}$ as a quantum algorithm which operates on $n$ qubits, each initially in the state $\ket{0}$, and whose quantum part finishes with a measurement of $k$ of the qubits in the computational basis. Given that the measurement returns outcome $x\in \{0,1\}^k$, the final output is then $\phi(x)$, for some fixed function $\phi:\{0,1\}^k \rightarrow \R$. If $\mathcal{A}$ is a classical randomised algorithm, or a quantum circuit using (for example) mixed states and intermediate measurements, a corresponding unitary quantum circuit of this form can be produced using standard reversible-computation techniques~\cite{aharonov98a}. As is common in works based on quantum amplitude amplification and estimation~\cite{brassard02}, we also assume that we have the ability to execute the algorithm $\mathcal{A}^{-1}$, which is the inverse of the unitary part of $\mathcal{A}$. If we do have a description of $\mathcal{A}$ as a quantum circuit, this can be achieved simply by running the circuit backwards, replacing each gate with its inverse.

We first deal with the special case where the output of $\mathcal{A}$ is bounded between 0 and 1. Here a quantum algorithm for approximating $\mu := \E[v(\mathcal{A})]$ quadratically faster than is possible classically can be found by combining ideas from previously known algorithms~\cite{heinrich01,brassard11,wocjan09}. We append an additional qubit and define a unitary operator $W$ on $k+1$ qubits which performs the map $\ket{x}\ket{0} \mapsto \ket{x}(\sqrt{1-\phi(x)}\ket{0} + \sqrt{\phi(x)}\ket{1})$. If the final measurement of the algorithm $\mathcal{A}$ is replaced with performing $W$, then measuring the added qubit, the probability that we receive the answer 1 is precisely $\mu$. Using quantum amplitude estimation~\cite{brassard02} the probability that this measurement returns 1 can be estimated to higher accuracy than is possible classically. Using $t$ iterations of amplitude estimation, we can output an estimate $\widetilde{\mu}$ such that $|\widetilde{\mu} - \mu|  = O(\sqrt{\mu} / t + 1 / t^2)$ with high probability~\cite{brassard02}. In particular, $O(1/\epsilon)$ iterations of amplitude estimation are sufficient to produce an estimate $\widetilde{\mu}$ such that $|\widetilde{\mu} - \mu| \le \epsilon$ with, say, 99\% probability.

The next step is to use the above algorithm as a subroutine in a more general procedure that can deal with algorithms $\mathcal{A}$ whose output is non-negative, has bounded $\ell_2$ norm, but is not necessarily bounded between 0 and 1. That is, algorithms for which we can control the expression $\left\|v(\mathcal{A})\right\|_2 := \sqrt{\E[v(\mathcal{A})^2]}$. The procedure for this case generalises, and is based on the same ideas as, a previously known result for the uniform distribution~\cite{heinrich01}.

The idea is to split the output of $\mathcal{A}$ up into disjoint intervals depending on size. Write $\mathcal{A}_{p,q}$ for the ``truncated'' algorithm which outputs $v(\mathcal{A})$ if $p \le v(\mathcal{A}) < q$, and otherwise outputs 0. We estimate $\mu$ by applying the above algorithm to estimate $\E[v(\mathcal{A}_{p,q})]$ for a sequence of $O(\log 1/\epsilon)$ intervals which are exponentially increasing in size, and summing the results. As the intervals $[p,q)$ get larger, the accuracy with which we approximate $\E[v(\mathcal{A}_{p,q})]$ decreases, and values $v(\mathcal{A})$ larger than about $1/\epsilon$ are ignored completely. However, the overall upper bound on $\left\|v(\mathcal{A})\right\|_2$ allows us to infer that these larger values do not affect the overall expectation $\mu$ much; indeed, if $\mu$ depended significantly on large values in the output, the $\ell_2$ norm of $v(\mathcal{A})$ would be high.

The final result is that for $\left\|v(\mathcal{A})\right\|_2 = O(1)$, given appropriate parameter choices, the estimate $\widetilde{\mu}$ satisfies $|\widetilde{\mu} - \mu| = O(\epsilon)$ with high probability, and the algorithm uses $\mathcal{A}$ $\widetilde{O}(1/\epsilon)$ times in total. This scaling is a near-quadratic improvement over the best possible classical algorithm. %The whole procedure is described in Figure \ref{alg:l2sketch} (with precise parameter choices given in Section~\ref{sec:algs}).

\begin{comment}
\boxalgmfig{alg:l2sketch}{Approximating the mean of positive functions with bounded $\ell_2$ norm (sketch)}{
{\bf Input:} an algorithm $\mathcal{A}$ with non-negative output, and an accuracy $\epsilon$.
\begin{enumerate}
\item Use $\widetilde{O}(1/\epsilon)$ iterations of amplitude estimation to estimate $\E[v(\mathcal{A}_{0,1})]$. Let the estimate be $\widetilde{\mu}_0$.
\item For each $\ell$ between 1 and $O(\log 1/\epsilon)$, use $\widetilde{O}(1/\epsilon)$ iterations of amplitude estimation to estimate $\E[v(\mathcal{A}_{2^{\ell-1},2^{\ell}})/2^{\ell}]$, and let the estimate be $\widetilde{\mu}_\ell$.
\item Output $\widetilde{\mu} = \widetilde{\mu}_0 + \sum_{\ell} 2^{\ell} \widetilde{\mu}_\ell$.
\end{enumerate}
}
\end{comment}

We next consider the more general case of algorithms $\mathcal{A}$ which have bounded variance, but whose output need not be non-negative, nor bounded in $\ell_2$ norm. To apply the previous algorithm, we would like to transform the output of $\mathcal{A}$ to make its $\ell_2$ norm low. If $v(\mathcal{A})$ has mean $\mu$ and variance upper-bounded by $\sigma^2$, a suitable way to achieve this is to subtract $\mu$ from the output of $\mathcal{A}$, then divide by $\sigma$. The new algorithm's output would have $\ell_2$ norm upper-bounded by 1, and estimating its expected value up to additive error $\epsilon / \sigma$ would give us an estimate of $\mu$ up to $\epsilon$. Unfortunately, we of course do not know $\mu$ initially, so cannot immediately implement this idea. To approximately implement it, we first run $\mathcal{A}$ once and use the output $\widetilde{m}$ as a proxy for $\mu$. Because $\Var(v(\mathcal{A})) \le \sigma^2$, $\widetilde{m}$ is quite likely to be within distance $O(\sigma)$ of $\mu$. Therefore, the algorithm $\mathcal{B}$ produced from $\mathcal{A}$ by subtracting $\widetilde{m}$ and dividing by $\sigma$ is quite likely to have $\ell_2$ norm upper-bounded by a constant. We can thus efficiently estimate the positive and negative parts of $\E[v(\mathcal{B})]$ separately, then combine and rescale them. The overall algorithm achieves accuracy $\epsilon$ in time $\widetilde{O}(\sigma / \epsilon)$.

A similar idea can be used to approximate the expected output value of algorithms for which we have a bound on the relative variance, namely that $\Var(v(\mathcal{A})) = O(\mu^2)$. In this setting it turns out that $\widetilde{O}(1/\epsilon)$ uses of $\mathcal{A}$ suffice to produce an estimate $\widetilde{\mu}$ accurate up to {\em relative} error $\epsilon$, i.e.\ for which $|\widetilde{\mu} - \mu| \le \epsilon\mu$. This is again a near-quadratic improvement over the best possible classical algorithm.

% ------------------------------------------------------------------------------

\subsection{Approximating partition functions}
\label{sec:approxpart}

In this section we discuss (with details in Section~\ref{sec:partition}) how these algorithms can be applied to the problem of approximating partition functions. Consider a (classical) physical system which has state space $\Omega$, together with a Hamiltonian $H:\Omega \rightarrow \R$ specifying the energy of each configuration\footnote{We use $x$ to label configurations rather than the more standard $\sigma$ to avoid confusion with the variance.} $x \in \Omega$. Here we will assume that $H$ takes integer values in the set $\{0,\dots,n\}$. A central problem is to compute the partition function
\[ Z(\beta) = \sum_{x \in \Omega} e^{-\beta\,H(x)} \]
for some inverse temperature $\beta$ defined by $\beta = 1 / (k_BT)$, where $T$ is the temperature and $k_B$ is Boltzmann's constant. As well as naturally encapsulating various models in statistical physics, such as the Ising and Potts models, this framework also encompasses well-studied problems in computer science, such as counting the number of valid $k$-colourings of a graph. In particular, $Z(\infty)$ counts the number of configurations $x$ such that $H(x) = 0$. It is often hard to compute $Z(\beta)$ for large $\beta$ but easy to approximate $Z(\beta) \approx |\Omega|$ for $\beta \approx 0$. In many cases, such as the Ising model, it is known that computing $Z(\infty)$ exactly falls into the \#P-complete complexity class~\cite{jerrum93}, and hence is unlikely to admit an efficient quantum or classical algorithm.

Here our goal will be to approximate $Z(\beta)$ up to relative error $\epsilon$, for some small $\epsilon$. That is, to output $\widetilde{Z}$ such that $|\widetilde{Z} - Z(\beta)| \le \epsilon\,Z(\beta)$, with high probability.
%An algorithm which achieves this goal and runs in time polynomial in the input size and $1/\epsilon$ is known as a fully polynomial-time approximation scheme (FPRAS)~\cite{jerrum03}.
For simplicity, we will focus on $\beta = \infty$ in the following discussion, but it is easy to see how to generalise to arbitrary $\beta$.

Let $0 = \beta_0 < \beta_1 < \dots < \beta_\ell = \infty$ be a sequence of inverse temperatures. A standard classical approach to design algorithms for approximating partition functions~\cite{valleau72,dyer92,bezakova08,stefankovic09,wocjan09} is based around expressing $Z(\beta_\ell)$ as the telescoping product
\[ Z(\beta_\ell) = Z(\beta_0) \frac{Z(\beta_1)}{Z(\beta_0)} \frac{Z(\beta_2)}{Z(\beta_1)} \dots \frac{Z(\beta_\ell)}{Z(\beta_{\ell-1})}. \]
If we can compute $Z(\beta_0) = |\Omega|$, and can also approximate each of the ratios $\alpha_i := Z(\beta_{i+1}) / Z(\beta_i)$ accurately, taking the product will give a good approximation to $Z(\beta_\ell)$. Let $\pi_i$ denote the Gibbs (or Boltzmann) probability distribution corresponding to inverse temperature $\beta_i$, where
\[ \pi_i(x) = \frac{1}{Z(\beta_i)} e^{-\beta_i H(x)}. \]
To approximate $\alpha_i$ we define the random variable
\[ Y_i(x) = e^{-(\beta_{i+1}-\beta_i) H(x) }. \] 
Then one can readily compute that $\E_{\pi_i}[Y_i] = \alpha_i$, so sampling from each distribution $\pi_i$ allows us to estimate the quantities $\alpha_i$. It will be possible to estimate $\alpha_i$ up to small relative error efficiently if the ratio $\E[Y_i^2]/\E[Y_i]^2$ is low. This motivates the concept of a {\em Chebyshev cooling schedule}~\cite{stefankovic09}: a sequence of inverse temperatures $\beta_i$ such that $\E[Y_i^2]/\E[Y_i]^2 = O(1)$ for all $i$. It is known that, for any partition function problem as defined above such that $|\Omega| = A$, there exists a Chebyshev cooling schedule with $\ell = \widetilde{O}(\sqrt{\log A})$~\cite{stefankovic09}. %Given such a schedule, we therefore obtain a classical algorithm which runs in time $\widetilde{O}((\log A)/\epsilon^2)$.

%We now show how to reduce this complexity using the quantum algorithm for approximating $\E[Y_i]$ up to low relative error.
It is sufficient to approximate $\E[Y_i]$ up to relative error $O(\epsilon / \ell)$ for each $i$ to get an overall approximation accurate up to relative error $\epsilon$. To achieve this, the quantum algorithm presented here needs to use at most $\widetilde{O}( \ell / \epsilon)$ samples from $Y_i$. Given a Chebyshev cooling schedule with $\ell = \widetilde{O}(\sqrt{\log A})$, the algorithm thus uses $\widetilde{O}((\log A)/\epsilon)$ samples in total, a near-quadratic improvement in terms of $\epsilon$ over the complexity of the fastest known classical algorithm~\cite{stefankovic09}.

In general, we cannot exactly sample from the distributions $\pi_i$. Classically, one way of approximately sampling from these distributions is to use a Markov chain which mixes rapidly and has stationary distribution $\pi_i$. For a reversible, ergodic Markov chain, the time required to produce such a sample is controlled by the {\em relaxation time} $\tau := 1/(1 - |\lambda_1|)$ of the chain, where $\lambda_1$ is the second largest eigenvalue in absolute value~\cite{levin09}. In particular, sampling from a distribution close to $\pi_i$ in total variation distance requires $\Omega(\tau)$ steps of the chain.

It has been known for some time that quantum walks can sometimes mix quadratically faster~\cite{aharonov01}. One case where efficient mixing can be obtained is for sequences of Markov chains whose stationary distributions $\pi$ are close~\cite{wocjan08}. Further, for this special case one can approximately produce coherent ``quantum sample'' states $\ket{\pi} = \sum_{x \in \Omega} \sqrt{\pi(x)} \ket{x}$ efficiently. Here we can show (Section \ref{sec:approxsamp}) that the Chebyshev cooling schedule condition implies that each distribution in the sequence $\pi_1,\dots,\pi_{\ell-1}$ is close enough to its predecessor that we can use techniques of~\cite{wocjan08} to approximately produce any state $\ket{\pi_i}$ using $\widetilde{O}(\ell \sqrt{\tau})$ quantum walk steps each. Using similar ideas we can approximately reflect about $\ket{\pi_i}$ using only $\widetilde{O}(\sqrt{\tau})$ quantum walk steps.

Approximating $\E[Y_i]$ up to relative error $O(\epsilon / \ell)$ using our algorithm requires one quantum sample approximating $\ket{\pi_i}$, and $\widetilde{O}(\ell/\epsilon)$ approximate reflections about $\ket{\pi_i}$. Therefore, the total number of quantum walk steps required for each $i$ is $\widetilde{O}(\ell \sqrt{\tau} / \epsilon)$. Summing over $i$, we get a quantum algorithm for approximating an arbitrary partition function up to relative error $\epsilon$ using $\widetilde{O}((\log A) \sqrt{\tau} / \epsilon)$ quantum walk steps. The fastest known classical algorithm~\cite{stefankovic09} exhibits quadratically worse dependence on both $\tau$ and $\epsilon$.

In the above discussion, we have neglected the complexity of computing the Chebyshev cooling schedule itself. An efficient classical algorithm for this task is known~\cite{stefankovic09}, which runs in time $\widetilde{O}((\log A) \tau)$. Adding the complexity of this part, we finish with an overall complexity of $\widetilde{O}((\log A) \sqrt{\tau} (\sqrt{\tau} + 1/ \epsilon))$. We leave the interesting question open of whether there exists a more efficient quantum algorithm for finding a Chebyshev cooling schedule.

% ------------------------------------------------------------------------------

\subsection{Applications}

We now sketch several representative settings (for details, see Section~\ref{sec:pfproblems}) in which our algorithm for approximating partition functions gives a quantum speedup.

\begin{itemize}
\item The {\bf ferromagnetic Ising model} above the critical temperature. This well-studied statistical physics model is defined in terms of a graph $G = (V,E)$ by the Hamiltonian $H(z) = -\sum_{(u,v) \in E} z_u z_v$, where $|V| = n$ and $z \in \{\pm 1\}^n$. The Markov chain known as the Glauber dynamics is known to mix rapidly above a certain critical temperature and to have as its stationary distribution the Gibbs distribution. For example, for any graph with maximum degree $O(1)$, the mixing time of the Glauber dynamics for sufficiently low inverse temperature $\beta$ is $O(n \log n)$~\cite{mossel13}. In this case, as $A = 2^n$, the quantum algorithm approximates $Z(\beta)$ to within relative error $\epsilon$ in $\widetilde{O}(n^{3/2}/ \epsilon + n^2)$ steps. The corresponding classical algorithm~\cite{stefankovic09} uses $\widetilde{O}(n^2 / \epsilon^2)$ steps.

\item {\bf Counting colourings.} Here we are given a graph $G$ with $n$ vertices and maximum degree $d$. We seek to approximately count the number of valid $k$-colourings of $G$, where a colouring of the vertices is valid if all pairs of neighbouring vertices are assigned different colours. In the case where $k>2d$, the use of a rapidly mixing Markov chain gives a quantum algorithm approximating the number of colourings of $G$ up to relative error $\epsilon$ in time $\widetilde{O}(n^{3/2}/ \epsilon + n^2)$, as compared with the classical $\widetilde{O}(n^2 / \epsilon^2)$~\cite{stefankovic09}.

\item {\bf Counting matchings.} A matching in a graph $G$ is a subset $M$ of the edges of $G$ such that no pair of edges in $M$ shares a vertex. In statistical physics, matchings are studied under the name of monomer-dimer coverings~\cite{heilmann72}. Our algorithm can approximately count the number of matchings on a graph with $n$ vertices and $m$ edges in $\widetilde{O}(n^{3/2} m^{1/2}/\epsilon +n^2m )$ steps, as compared with the classical $\widetilde{O}(n^2 m / \epsilon^2)$~\cite{stefankovic09}.
\end{itemize}

Finally, as another example of how our algorithm can be applied, we improve the accuracy of an existing quantum algorithm for estimating the total variation distance between probability distributions. In this setting, we are given the ability to sample from probability distributions $p$ and $q$ on $n$ elements, and would like to estimate the distance between them up to additive error $\epsilon$. A quantum algorithm of Bravyi, Harrow and Hassidim solves this problem using $O(\sqrt{n}/\epsilon^8)$ samples~\cite{bravyi11a}, while no classical algorithm can achieve sublinear dependence on $n$~\cite{valiant11}.

Quantum mean estimation can significantly improve the dependence of this quantum algorithm on $\epsilon$. The total variation distance between $p$ and $q$ can be described as the expected value of the random variable $R(x) = \frac{ | p(x) - q(x) |}{p(x) + q(x)}$, where $x$ is drawn from the distribution $r = (p+q)/2$~\cite{bravyi11a}. For each $x$, $R(x)$ can be computed up to accuracy $\epsilon$ using $\widetilde{O}(\sqrt{n}/\epsilon^{3/2})$ iterations of amplitude estimation. Wrapping this within $O(1/\epsilon)$ iterations of the mean-estimation algorithm, we obtain an overall algorithm running in time $\widetilde{O}(\sqrt{n}/\epsilon^{5/2})$. See Section~\ref{sec:tvd} for details.

%In the rest of the paper, we provide rigorous statements and details for the above discussion.

%\putbib[../../thesis.bib]
%\end{bibunit}

% ------------------------------------------------------------------------------

%\begin{bibunit}[plain]

\section{Algorithms}
\label{sec:algs}

We now give technical details, parameter values and proofs for the various algorithms described informally in Section \ref{sec:approx}. Recall that, for any randomised or quantum algorithm $\mathcal{A}$, we let $v(\mathcal{A})$ be the random variable corresponding to the value computed by $\mathcal{A}$. We assume that $\mathcal{A}$ takes no input directly, but may have access to input (e.g.\ via queries to some black box or ``oracle'') during its execution. We further assume throughout that $\mathcal{A}$ is a quantum algorithm of the following form: apply some unitary operator to the initial state $\ket{0^n}$; measure $k \le n$ qubits of the resulting state in the computational basis, obtaining outcome $x \in \{0,1\}^k$; output $\phi(x)$ for some easily computable function $\phi:\{0,1\}^k \rightarrow \R$. We finally assume that we have access to the inverse of the unitary part of the algorithm, which we write as $\mathcal{A}^{-1}$.

\begin{lem}[Powering lemma~\cite{jerrum86}]
\label{lem:powering}
Let $\mathcal{A}$ be a (classical or quantum) algorithm which aims to estimate some quantity $\mu$, and whose output $\widetilde{\mu}$ satisfies $|\mu - \widetilde{\mu}| \le \epsilon$ except with probability $\gamma$, for some fixed $\gamma < 1/2$. Then, for any $\delta > 0$, it suffices to repeat $\mathcal{A}$ $O(\log 1/\delta)$ times and take the median to obtain an estimate which is accurate to within $\epsilon$ with probability at least $1-\delta$.
\end{lem}

We will also need the following fundamental result from~\cite{brassard02}:

%\begin{thm}[Amplitude estimation~\cite{brassard02}]
%\label{thm:ampest}
%There is a quantum algorithm called {\bf amplitude estimation} which takes as input one copy of a quantum state $\ket{\psi}$, a unitary transformation $U = 2\proj{\psi} - I$, a unitary transformation $V = I - 2 P$ for some projector $P$, and an integer $t$. The algorithm outputs $\widetilde{a}$, an estimate of $a = \bracket{\psi}{P}{\psi}$, such that
%
%\[ |\widetilde{a}-a| \le 2\pi \frac{\sqrt{a(1-a)}}{t} + \frac{\pi^2}{t^2} \]
%
%with probability at least $8/\pi^2$, using $U$ and $V$ $t$ times each. In particular, for any $\epsilon$ such that $0 \le \epsilon \le 1$, to produce an estimate $\widetilde{a}$ such that, with probability at least $8/\pi^2$:
%\begin{itemize}
%\item $| \widetilde{a} - a | \le \epsilon a$: it suffices to take $t = \lceil 4\pi/(\epsilon \sqrt{a})\rceil$;
%\item $| \widetilde{a} - a | \le \epsilon$: it suffices to take $t = \lceil \pi/(2\epsilon) \rceil$.
%\end{itemize}
%The input state $\ket{\psi}$ is left unchanged by the algorithm.
%\end{thm}

\begin{thm}[Amplitude estimation~\cite{brassard02}]
\label{thm:ampest}
There is a quantum algorithm called {\bf amplitude estimation} which takes as input one copy of a quantum state $\ket{\psi}$, a unitary transformation $U = 2\proj{\psi} - I$, a unitary transformation $V = I - 2 P$ for some projector $P$, and an integer $t$. The algorithm outputs $\widetilde{a}$, an estimate of $a = \bracket{\psi}{P}{\psi}$, such that
\[ |\widetilde{a}-a| \le 2\pi \frac{\sqrt{a(1-a)}}{t} + \frac{\pi^2}{t^2} \]
with probability at least $8/\pi^2$, using $U$ and $V$ $t$ times each.
\end{thm}

The success probability of $8/\pi^2$ can be improved to $1-\delta$ for any $\delta > 0$ using the powering lemma at the cost of an $O(\log 1/\delta)$ multiplicative factor.

% ------------------------------------------------------------------------------

\subsection{Estimating the mean with bounded output values}
\label{sec:bounded}

We first consider the problem of estimating $\E[v(\mathcal{A})]$ in the special case where $v(\mathcal{A})$ is bounded between 0 and 1. The algorithm for this case is effectively a combination of elegant ideas of Brassard et al.~\cite{brassard11} and Wocjan et al.~\cite{wocjan09}. The former described an algorithm for efficiently approximating the mean of an arbitrary function with respect to the uniform distribution; the latter described an algorithm for approximating the expected value of a particular observable, with respect to an arbitrary quantum state. The first quantum algorithm achieving optimal scaling for approximating the mean of a bounded function under the uniform distribution was due to Heinrich~\cite{heinrich01}.

\boxalgm{alg:meanbounded}{Approximating the mean output value of algorithms bounded between 0 and 1 (cf.~\cite{brassard11,heinrich01,wocjan09})}{
{\bf Input:} an algorithm $\mathcal{A}$ such that $0 \le v(\mathcal{A}) \le 1$, integer $t$, real $\delta > 0$.\\ Assume that $\mathcal{A}$ is a quantum algorithm which makes no measurements until the end of the algorithm; operates on initial input state $\ket{0^n}$; and its final measurement is a measurement of the last $k \le n$ of these qubits in the computational basis.
\begin{enumerate}
\item If necessary, modify $\mathcal{A}$ such that it makes no measurements until the end of the algorithm; operates on initial input state $\ket{0^n}$; and its final measurement is a measurement of the last $k \le n$ of these qubits in the computational basis.
\item Let $W$ be the unitary operator on $k+1$ qubits defined by
\[ W\ket{x}\ket{0} = \ket{x} \left(\sqrt{1-\phi(x)}\ket{0} + \sqrt{\phi(x)}\ket{1}\right), \]
where each computational basis state $x \in \{0,1\}^k$ is associated with a real number $\phi(x) \in [0,1]$ such that $\phi(x)$ is the value output by $\mathcal{A}$ when measurement outcome $x$ is received.
\item Repeat the following step $O(\log 1/\delta)$ times and output the median of the results:
\begin{enumerate}
\item Apply $t$ iterations of amplitude estimation, setting $\ket{\psi} = (I \otimes W)(\mathcal{A} \otimes I)\ket{0^{n+1}}$, $P = I \otimes \proj{1}$.
\end{enumerate}
\end{enumerate}
}

\begin{thm}
\label{thm:meanerr}
Let $\ket{\psi}$ be defined as in Algorithm \ref{alg:meanbounded} and set $U = 2\proj{\psi} - I$. Algorithm \ref{alg:meanbounded} uses $O(\log 1/\delta)$ copies of the state $\mathcal{A} \ket{0^n}$, uses $U$ $O(t \log 1/\delta)$ times, and outputs an estimate $\widetilde{\mu}$ such that
\[ | \widetilde{\mu} - \E[v(\mathcal{A})] | \le C \left( \frac{\sqrt{\E[v(\mathcal{A})]}}{t} + \frac{1}{t^2} \right) \]
with probability at least $1-\delta$, where $C$ is a universal constant. In particular, for any fixed $\delta > 0$ and any $\epsilon$ such that $0 \le \epsilon \le 1$, to produce an estimate $\widetilde{\mu}$ such that with probability at least $1-\delta$, $| \widetilde{\mu} - \E[v(\mathcal{A})] | \le \epsilon\,\E[v(\mathcal{A})]$ it suffices to take $t = O(1/(\epsilon \sqrt{\E[v(\mathcal{A})]}))$. To achieve $| \widetilde{\mu} - \E[v(\mathcal{A})] | \le \epsilon$ with probability at least $1-\delta$ it suffices to take $t = O(1/\epsilon)$.
%\begin{itemize}
%\item $| \widetilde{\mu} - \E[v(\mathcal{A})] | \le \epsilon\,\E[v(\mathcal{A})]$: it suffices to take $t = O(1/(\epsilon \sqrt{\E[v(\mathcal{A})]}))$;
%\item $| \widetilde{\mu} - \E[v(\mathcal{A})] | \le \epsilon$: it suffices to take $t = O(1/\epsilon)$.
%\end{itemize}
\end{thm}

\begin{proof}
The complexity claim follows immediately from Theorem \ref{thm:ampest}. Also observe that $W$ can be implemented efficiently, as it is a controlled rotation of one qubit dependent on the value of $\phi(x)$~\cite{wocjan09}. It remains to show the accuracy claim. The final state of $\mathcal{A}$, just before its last measurement, can be written as
\[ \ket{\psi'} = \mathcal{A}\ket{0^n} = \sum_x \alpha_x \ket{\psi_x} \ket{x} \]
for some normalised states $\ket{\psi_x}$. If we then attach an ancilla qubit and apply $W$, we obtain
\[ \ket{\psi} = (I \otimes W) (\mathcal{A} \otimes I) \ket{0^n}\ket{0} = \sum_x \alpha_x \ket{\psi_x} \ket{x} \left(\sqrt{1-\phi(x)}\ket{0} + \sqrt{\phi(x)}\ket{1}\right). \]
We have
\[ \bracket{\psi}{P}{\psi} = \sum_x |\alpha_x|^2 \phi(x) = \E[v(\mathcal{A})]. \]
Therefore, when we apply amplitude estimation, by Theorem \ref{thm:ampest} we obtain an estimate $\widetilde{\mu}$ of $\mu = \E[v(\mathcal{A})]$ such that 
\[ |\widetilde{\mu}-\mu| \le 2\pi \frac{\sqrt{\mu(1-\mu)}}{t} + \frac{\pi^2}{t^2} \]
with probability at least $8/\pi^2$. The powering lemma (Lemma \ref{lem:powering}) implies that the median of $O(\log 1/\delta)$ repetitions will lie within this accuracy bound with probability at least $1-\delta$.
\end{proof}

Observe that $U = 2\proj{\psi} - I$ can be implemented with one use each of $\mathcal{A}$ and $\mathcal{A}^{-1}$, and $V = I-2P$ is easy to implement.

%Theorem \ref{thm:meanerr} is stated in terms of uses of the operator $U$, i.e.\ reflections about the state $\ket{\psi} = \mathcal{A}\ket{0}$. Such a reflection can be implemented with one use of $\mathcal{A}$ applied to the initial state $\ket{0}$, and one use of $\mathcal{A}^{-1}$. However, in some cases (as we will see below) it can be more costly to create $\ket{\psi}$ than to implement $U$.

It seems likely that the median-finding algorithm of Nayak and Wu~\cite{nayak99} could also be generalised in a similar way, to efficiently compute the median of the output values of any quantum algorithm. As we will not need this result here we do not pursue this further.

% ------------------------------------------------------------------------------

\subsection{Estimating the mean with bounded $\ell_2$ norm}
\label{sec:boundedl2}

We now use Algorithm \ref{alg:meanbounded} to give an efficient quantum algorithm for approximating the mean output value of a quantum algorithm whose output has bounded $\ell_2$ norm. In what follows, for any algorithm $\mathcal{A}$, let $\mathcal{A}_{< x}$, $\mathcal{A}_{x,y}$, $\mathcal{A}_{\ge y}$, be the algorithms defined by executing $\mathcal{A}$ to produce a value $v(\mathcal{A})$ and:
\begin{itemize}
\item $\mathcal{A}_{<x}$: If $v(\mathcal{A}) < x$, output $v(\mathcal{A})$, otherwise output 0;
\item $\mathcal{A}_{x,y}$: If $x \le v(\mathcal{A}) < y$, output $v(\mathcal{A})$, otherwise output 0;
\item $\mathcal{A}_{\ge y}$: If $y \le v(\mathcal{A})$, output $v(\mathcal{A})$, otherwise output 0.
\end{itemize}
In addition, for any algorithm $\mathcal{A}$ and any function $f:\R \rightarrow \R$, let $f(\mathcal{A})$ be the algorithm produced by evaluating $v(\mathcal{A})$ and computing $f(v(\mathcal{A}))$. Note that Algorithm \ref{alg:meanbounded} can easily be modified to compute $\E[f(v(\mathcal{A}))]$ rather than $\E[v(\mathcal{A})]$, for any function $f:\R \rightarrow [0,1]$, by modifying the operation $W$.

Our algorithm and correctness proof are a generalisation of a result of Heinrich~\cite{heinrich01} for computing the mean with respect to the uniform distribution of functions with bounded $L^2$ norm, and are based on the same ideas. Write $\left\|v(\mathcal{A})\right\|_2 := \sqrt{\E[v(\mathcal{A})^2]}$.

\boxalgm{alg:meansub}{Approximating the mean of positive functions with bounded $\ell_2$ norm}{
{\bf Input:} an algorithm $\mathcal{A}$ such that $v(\mathcal{A}) \ge 0$, and an accuracy $\epsilon < 1/2$.
\begin{enumerate}
\item Set $k = \lceil\log_2 1/\epsilon \rceil$, $t_0 = \left\lceil\frac{D\sqrt{\log_2 1/\epsilon}}{\epsilon}\right\rceil$, where $D$ is a universal constant to be chosen later.
\item Use Algorithm \ref{alg:meanbounded} with $t= t_0$, $\delta = 1/10$ to estimate $\E[v(\mathcal{A}_{0,1})]$. Let the estimate be $\widetilde{\mu}_0$.
\item For $\ell=1,\dots, k$:
\begin{enumerate}
\item Use Algorithm \ref{alg:meanbounded} with $t= t_0$, $\delta = 1/(10k)$ to estimate $\E[v(\mathcal{A}_{2^{\ell-1},2^{\ell}})/2^{\ell}]$. Let the estimate be $\widetilde{\mu}_\ell$.
\end{enumerate}
\item Output $\widetilde{\mu} = \widetilde{\mu}_0 + \sum_{\ell=1}^k 2^{\ell} \widetilde{\mu}_\ell$.
\end{enumerate}
}

\begin{lem}
\label{lem:boundedl2}
Let $\ket{\psi} = \mathcal{A} \ket{0^n}$, $U = 2\proj{\psi} - I$. Algorithm \ref{alg:meansub} uses $O(\log(1/\epsilon)\log \log(1/\epsilon))$ copies of $\ket{\psi}$, uses $U$ $O((1/\epsilon)\log^{3/2}(1/\epsilon)\log \log(1/\epsilon))$ times, and estimates $\E[v(\mathcal{A})]$ up to additive error $\epsilon(\left\|v(\mathcal{A})\right\|_2+1)^2$ with probability at least $4/5$.
\end{lem}

\begin{proof}
We first show the resource bounds. Algorithm \ref{alg:meanbounded} is run $\Theta(\log 1/\epsilon)$ times, each time with parameter $\delta = \Omega(1/(\log 1/\epsilon))$. By Theorem \ref{thm:meanerr}, each use of Algorithm \ref{alg:meanbounded} consumes $O(\log \log 1/\epsilon)$ copies of $\ket{\psi}$ and uses $U$ $O((1/\epsilon)\sqrt{\log(1/\epsilon)} \log \log(1/\epsilon) )$ times. The total number of copies of $\ket{\psi}$ used is $O(\log(1/\epsilon)\log \log(1/\epsilon))$, and the total number of uses of $U$ is $O((1/\epsilon)\log^{3/2}(1/\epsilon)\log \log(1/\epsilon))$.

All of the uses of Algorithm \ref{alg:meanbounded} succeed, except with probability at most $1/5$ in total. To estimate the total error in the case where they all succeed, we write
\[ \E[v(\mathcal{A})] = \E[v(\mathcal{A}_{0,1})] + \sum_{\ell=1}^k 2^\ell \E[v(\mathcal{A}_{2^{\ell-1},2^{\ell}})/2^{\ell}] + \E[v(\mathcal{A}_{\ge2^k})] \]
and use the triangle inequality term by term to obtain
\[ |\widetilde{\mu} - \E[v(\mathcal{A})]| \le |\widetilde{\mu}_0 - \E[v(\mathcal{A}_{0,1})]| + \sum_{\ell=1}^k 2^{\ell} | \widetilde{\mu}_\ell - \E[v(\mathcal{A}_{2^{\ell-1},2^{\ell}})/2^\ell]| + \E[v(\mathcal{A}_{\ge2^k})]. \]
Let $p(x)$ denote the probability that $\mathcal{A}$ outputs $x$. We have
\[ \E[v(\mathcal{A}_{\ge 2^k})] = \sum_{x \ge 2^k}  p(x) x \le \frac{1}{2^k} \sum_x p(x) x^2 = \frac{\left\|v(\mathcal{A})\right\|_2^2}{2^k}. \]
By Theorem \ref{thm:meanerr},
\[ | \widetilde{\mu}_0 - \E[v(\mathcal{A}_{0,1})]| \le C \left( \frac{\sqrt{\E[v(\mathcal{A}_{0,1})]}}{t_0} + \frac{1}{t_0^2} \right) \]
and similarly
\[ | \widetilde{\mu}_\ell - \E[v(\mathcal{A}_{2^{\ell-1},2^{\ell}})/2^\ell]| \le C \left( \frac{\sqrt{\E[v(\mathcal{A}_{2^{\ell-1},2^{\ell}})]}}{t_0\,2^{\ell/2}} + \frac{1}{t_0^2} \right). \]
So the total error is at most
\[ C \left(\frac{\sqrt{\E[v(\mathcal{A}_{0,1})]}}{t_0} + \frac{1}{t_0^2}  + \sum_{\ell=1}^k 2^\ell \left( \frac{\sqrt{\E[v(\mathcal{A}_{2^{\ell-1},2^{\ell}})]}}{t_0\,2^{\ell/2}} + \frac{1}{t_0^2} \right) \right) + \frac{\left\|v(\mathcal{A})\right\|_2^2}{2^k}. \]
We apply Cauchy-Schwarz to the first part of each term in the sum:
\[ \sum_{\ell=1}^k 2^{\ell/2} \sqrt{\E[v(\mathcal{A}_{2^{\ell-1},2^{\ell}})]} \le \sqrt{k} \left(\sum_{\ell=1}^k 2^\ell \E[v(\mathcal{A}_{2^{\ell-1},2^{\ell}})] \right)^{1/2} \le \sqrt{2k} \left\|v(\mathcal{A})\right\|_2, \]
where the second inequality follows from
\[ \E[v(\mathcal{A}_{2^{\ell-1},2^\ell})] = \sum_{2^{\ell-1} \le x < 2^\ell}  p(x) x \le \frac{1}{2^{\ell-1}} \sum_{2^{\ell-1} \le x < 2^\ell}  p(x) x^2 = \frac{\|v(\mathcal{A}_{2^{\ell-1},2^{\ell}})\|_2^2}{2^{\ell-1}}. \]
Inserting this bound and using $\E[v(\mathcal{A}_{0,1})] \le 1$, we obtain
\[ |\widetilde{\mu} - \E[v(\mathcal{A})]| \le C \left(\frac{1}{t_0} + \frac{1}{t_0^2} + \frac{\sqrt{2k} \left\|v(\mathcal{A})\right\|_2}{t_0} + \frac{2^{k+1}}{t_0^2} \right) + \frac{\left\|v(\mathcal{A})\right\|_2^2}{2^k}. \]
Inserting the definitions of $t_0$ and $k$, we get an overall error bound
\begin{align*}
 & \!\!\!\!\!\!\!\! |\widetilde{\mu} - \E[v(\mathcal{A})]|\\
 &\le \frac{C}{D} \left(\frac{\epsilon}{\sqrt{\log_2 1/\epsilon}} + \frac{\epsilon^2}{D\log_2 1/\epsilon} + \sqrt{2}\epsilon\left\|v(\mathcal{A})\right\|_2 \left(1 + \frac{1}{\log_2 1/\epsilon} \right)^{1/2} + \frac{4\epsilon}{D \log_2 1/\epsilon} \right) + \epsilon \left\|v(\mathcal{A})\right\|_2^2\\
&\le \frac{C}{D} \left( \epsilon + \frac{\epsilon}{D} + 2\epsilon \left\|v(\mathcal{A})\right\|_2 + \frac{4\epsilon}{D} \right) + \epsilon \left\|v(\mathcal{A})\right\|_2^2\\
&= \epsilon\left( \frac{C}{D}\left(1+\frac{5}{D} + 2\left\|v(\mathcal{A})\right\|_2\right) + \left\|v(\mathcal{A})\right\|_2^2\right)
%&\le& C \left(\frac{\epsilon}{D} + \frac{\epsilon^2}{D^2} + \frac{\sqrt{2} \epsilon\left\|v(\mathcal{A})\right\|_2 }{D} \left(1 + \frac{2}{\log_2 1/\epsilon} \right)^{1/2} +  \frac{8\epsilon}{D^2 \log_2 1/\epsilon} \right) + \frac{\epsilon\left\|v(\mathcal{A})\right\|_2^2 }{2}.
\end{align*}
using $0<\epsilon < 1/2$ in the second inequality. For a sufficiently large constant $D$, this is upper-bounded by $\epsilon(\left\|v(\mathcal{A})\right\|_2+1)^2$ as claimed.
\end{proof}

Observe that, if $\E[v(\mathcal{A})^2] = O(1)$, to achieve additive error $\epsilon$ the number of uses of $\mathcal{A}$ that we need is $O((1/\epsilon)\log^{3/2}(1/\epsilon)\log \log(1/\epsilon))$. By the powering lemma, we can repeat Algorithm \ref{alg:meansub} $O(\log 1/\delta)$ times and take the median to improve the probability of success to $1-\delta$ for any $\delta>0$.

% ------------------------------------------------------------------------------

\subsection{Estimating the mean with bounded variance}

We are now ready to formally state our algorithm for estimating the mean output value of an arbitrary algorithm with bounded variance. For clarity, some of the steps are reordered as compared with the informal description in Section \ref{sec:approx}. Recall that, in the classical setting, if we wish to estimate $\E[v(\mathcal{A})]$ up to additive error $\epsilon$ for an arbitrary algorithm $\mathcal{A}$ such that
\[ \Var(v(\mathcal{A})) := \E[(v(\mathcal{A}) - \E[v(\mathcal{A})])^2] \le \sigma^2, \]
we need to use $\mathcal{A}$ $\Omega(\sigma^2 / \epsilon^2)$ times~\cite{dagum00}.

\boxalgm{alg:meanvariance}{Approximating the mean with bounded variance}{
{\bf Input:} an algorithm $\mathcal{A}$ such that $\Var(v(\mathcal{A})) \le \sigma^2$ for some known $\sigma$, and an accuracy $\epsilon$ such that $\epsilon < 4 \sigma$.
\begin{enumerate}
\item Set $\mathcal{A}' = \mathcal{A} / \sigma$.
\item Run $\mathcal{A}'$ once and let $\widetilde{m}$ be the output.
\item Let $\mathcal{B}$ be the algorithm produced by executing $\mathcal{A}'$ and subtracting $\widetilde{m}$.
\item Apply Algorithm \ref{alg:meansub} to algorithms $-\mathcal{B}_{<0}/4$ and $\mathcal{B}_{\ge 0}/4$ with accuracy $\epsilon/(32\sigma)$ and failure probability $1/9$, to produce estimates $\widetilde{\mu}^-$, $\widetilde{\mu}^+$ of $\E[v(-\mathcal{B}_{<0})/4]$ and $\E[v(\mathcal{B}_{\ge 0})/4]$, respectively.
\item Set $\widetilde{\mu} = \widetilde{m} - 4\widetilde{\mu}^- + 4\widetilde{\mu}^+$.
\item Output $\sigma\, \widetilde{\mu}$.
\end{enumerate}
}

\begin{thm}
Let $\ket{\psi} = \mathcal{A} \ket{0^n}$, $U = 2 \proj{\psi} - I$. Algorithm \ref{alg:meanvariance} uses $O(\log(\sigma/\epsilon) \log \log (\sigma/\epsilon))$ copies of $\ket{\psi}$, uses $U$ $O((\sigma/ \epsilon) \log^{3/2} (\sigma/\epsilon) \log \log (\sigma/\epsilon))$ times, and estimates $\E[v(\mathcal{A})]$ up to additive error $\epsilon$ with success probability at least $2/3$.
\end{thm}

\begin{proof}
First, observe that $\widetilde{m}$ is quite close to $\mu' := \E[v(\mathcal{A}')]$ with quite high probability. As $\Var(v(\mathcal{A}')) = \Var(v(\mathcal{A}))/\sigma^2 \le 1$, by Chebyshev's inequality we have
\[ \Pr[ |v(\mathcal{A}') - \mu'| \ge 3 ] \le \frac{1}{9}. \]
We therefore assume that $|\widetilde{m} - \mu'| \le 3$. In this case we have
\beas \|v(\mathcal{B})\|_2 &=& \E[v(\mathcal{B})^2]^{1/2} = \E[((v(\mathcal{A}') - \mu') + (\mu' - \widetilde{m}))^2]^{1/2}\\
&\le& \E[(v(\mathcal{A}') - \mu')^2]^{1/2} + \E[(\mu' - \widetilde{m})^2]^{1/2}\\
&\le& 4,
\eeas
where the first inequality is the triangle inequality. Thus $\|v(\mathcal{B})/4\|_2 \le 1$, which implies that $\|v(-\mathcal{B}_{<0})/4\|_2 \le 1$ and $\|v(\mathcal{B}_{\ge 0})/4\|_2 \le 1$.

The next step is to use Algorithm \ref{alg:meansub} to estimate $\E[v(-\mathcal{B}_{<0})/4]$ and $\E[v(\mathcal{B}_{\ge 0})/4]$ with accuracy $\epsilon/(32\sigma)$ and failure probability $1/9$. By Lemma \ref{lem:boundedl2}, if the algorithm succeeds in both cases the estimates are accurate up to $\epsilon / (8\sigma)$. We therefore obtain an approximation of each of $\E[v(-\mathcal{B}_{<0})]$ and $\E[v(\mathcal{B}_{\ge 0})]$ up to additive error $\epsilon/(2\sigma)$. As we have
\[ \E[v(\mathcal{A})] = \sigma\,\E[v(\mathcal{A}')] = \sigma(\widetilde{m} - \E[v(-\mathcal{B}_{<0})] + \E[v(\mathcal{B}_{\ge 0})]) \]
by linearity of expectation, using a union bound we have that $\sigma\,\widetilde{\mu}$ approximates $\E[v(\mathcal{A})]$ up to additive error $\epsilon$ with probability at least $2/3$.
\end{proof}

% ------------------------------------------------------------------------------

\subsection{Estimating the mean with bounded relative error}

It is often useful to obtain an estimate of the mean output value of an algorithm which is accurate up to small relative error, rather than the absolute error achieved by Algorithm \ref{alg:meanvariance}. Assume that we have the bound on the relative variance that $\Var(v(\mathcal{A}))/(\E[v(\mathcal{A})])^2 \le B$, where we normally think of $B$ as small, e.g.\ $B = O(1)$. Classically, it follows from Chebyshev's inequality that the simple classical algorithm described in the Introduction approximates $\E[v(\mathcal{A})]$ up to additive error $\epsilon\,\E[v(\mathcal{A})]$ with $O(B / \epsilon^2)$ uses of $\mathcal{A}$. In the quantum setting, we can improve the dependence on $\epsilon$ near-quadratically.

\boxalgm{alg:meanratio}{Approximating the mean with bounded relative error}{
{\bf Input:} An algorithm $\mathcal{A}$ such that $v(\mathcal{A}) \ge 0$ and $\Var(v(\mathcal{A}))/(\E[v(\mathcal{A})])^2 \le B$ for some $B \ge 1$, and an accuracy $\epsilon < 27B/4$.
\begin{enumerate}
\item Run $\mathcal{A}$ $k = \lceil 32 B \rceil$ times, receiving output values $v_1,\dots,v_k$, and set $\widetilde{m} = \frac{1}{k} \sum_{i=1}^k v_i$.
\item Apply Algorithm \ref{alg:meansub} to $\mathcal{A}/\widetilde{m}$ with accuracy $2 \epsilon / (3 (2 \sqrt{B} + 1)^2)$ and failure probability $1/8$. Let $\widetilde{\mu}$ be the output of the algorithm, multiplied by $\widetilde{m}$.
\item Output $\widetilde{\mu}$.
\end{enumerate}
}

\begin{thm}
\label{thm:meanratio}
Let $\ket{\psi} = \mathcal{A} \ket{0^n}$, $U = 2 \proj{\psi} - I$. Algorithm \ref{alg:meanratio} uses $O(B + \log(1/\epsilon) \log \log (1/\epsilon))$ copies of $\ket{\psi}$, uses $U$ $O((B/ \epsilon) \log^{3/2} (B/\epsilon) \log \log (B/\epsilon))$ times, and outputs an estimate $\widetilde{\mu}$ such that
\[ \Pr[|\widetilde{\mu} - \E[v(\mathcal{A})]| \ge \epsilon\,\E[v(\mathcal{A})]] \le 1/4. \]
\end{thm}

\begin{proof}
The complexity bounds follow from Lemma \ref{lem:boundedl2}; we now analyse the claim about accuracy. $\widetilde{m}$ is a random variable whose expectation is $\E[v(\mathcal{A})]$ and whose variance is $\Var(v(\mathcal{A}))/\lceil 32 B \rceil$. By Chebyshev's inequality, we have
\[ \Pr[|\widetilde{m}-\E[\widetilde{m}]| \ge |\E[\widetilde{m}]|/2 ] \le \frac{4 \Var(\widetilde{m})}{\E[\widetilde{m}]^2} = \frac{4 \Var(v(\mathcal{A}))}{\lceil 32 B \rceil \E[v(\mathcal{A})]^2} \le \frac{1}{8}. \]
We can thus assume that $\E[v(\mathcal{A})]/2 \le \widetilde{m} \le 3\,\E[v(\mathcal{A})]/2$. In this case, when we apply Algorithm \ref{alg:meansub} to $\mathcal{A}/\widetilde{m}$, we receive an estimate of $\E[v(\mathcal{A})]/\widetilde{m}$ which is accurate up to additive error
\[ \frac{ 2 \epsilon(\left\|v(\mathcal{A})\right\|_2 / \widetilde{m} + 1)^2}{3 (2 \sqrt{B} + 1)^2} \le \frac{ \epsilon\,\E[v(\mathcal{A})] (2 \left\|v(\mathcal{A})\right\|_2 / \E[v(\mathcal{A})] + 1)^2}{\widetilde{m}(2 \sqrt{B} + 1)^2} \le \frac{\epsilon\,\E[v(\mathcal{A})]}{\widetilde{m}} \]
except with probability $1/8$, where we use $\left\|v(\mathcal{A})\right\|_2/\E[v(\mathcal{A})] \le \sqrt{B}$. Multiplying by $\widetilde{m}$ and taking a union bound, we get an estimate of $\E[v(\mathcal{A})]$ which is accurate up to $\epsilon$ except with probability at most $1/4$.
\end{proof}

Once again, using the powering lemma we can repeat Algorithms \ref{alg:meanvariance} and \ref{alg:meanratio} $O(\log 1/\delta)$ times and take the median to improve their probabilities of success to $1-\delta$ for any $\delta>0$.

To see that Algorithms \ref{alg:meanvariance} and \ref{alg:meanratio} are close to optimal, we can appeal to a result of Nayak and Wu~\cite{nayak99}. Let $\mathcal{A}$ be an algorithm which picks an integer $x$ between 1 and $N$ uniformly at random, for some large $N$, and outputs $f(x)$ for some function $f:\{1,\dots,N\} \rightarrow \{0,1\}$. Then $\E[v(\mathcal{A})] = |\{x:f(x)=1\}|/N$. It was shown by Nayak and Wu~\cite{nayak99} that any quantum algorithm which computes this quantity for an arbitrary function $f$ up to (absolute or relative) error $\epsilon$ must make at most $\Omega(1/\epsilon)$ queries to $f$ in the case that $|\{x:f(x)=1\}|=N/2$. As the output of $\mathcal{A}$ for any such function has variance $1/4$, this implies that Algorithms \ref{alg:meansub} and \ref{alg:meanratio} are optimal in the black-box setting in terms of their scaling with $\epsilon$, up to polylogarithmic factors. By rescaling, we get a similar near-optimality claim for Algorithm \ref{alg:meanvariance} in terms of its scaling with $\sigma$.

% ------------------------------------------------------------------------------

\section{Partition function problems}
\label{sec:partition}

In this section we formally state and prove our results about partition function problems. We first recall the definitions from Section \ref{sec:approxpart}. A partition function $Z$ is defined by
\[ Z(\beta) = \sum_{x \in \Omega} e^{-\beta\,H(x)} \]
where $\beta$ is an inverse temperature and $H$ is a Hamiltonian function taking integer values in the set $\{0,\dots,n\}$. Let $0 = \beta_0 < \beta_1 < \dots < \beta_\ell = \infty$ be a sequence of inverse temperatures and assume that we can easily compute $Z(\beta_0) = |\Omega|$. We want to approximate $Z(\infty)$ by approximating the ratios $\alpha_i := Z(\beta_{i+1}) / Z(\beta_i)$ and using the telescoping product
\[ Z(\beta_\ell) = Z(\beta_0) \frac{Z(\beta_1)}{Z(\beta_0)} \frac{Z(\beta_2)}{Z(\beta_1)} \dots \frac{Z(\beta_\ell)}{Z(\beta_{\ell-1})}. \]
Finally, a sequence of Gibbs distributions $\pi_i$ is defined by
\[ \pi_i(x) = \frac{1}{Z(\beta_i)} e^{-\beta_i H(x)}. \]

% ------------------------------------------------------------------------------

\subsection{Chebyshev cooling schedules}
\label{app:chebyshev}

We start by motivating, and formally defining, the concept of a Chebyshev cooling schedule~\cite{stefankovic09}. To approximate $\alpha_i$ we define the random variable
\[ Y_i(x) = e^{-(\beta_{i+1}-\beta_i) H(x) }. \] 
Then
\[ \E[Y_i] := \E_{\pi_i}[Y_i] = \frac{1}{Z(\beta_i)} \sum_{x \in \Omega} e^{-\beta_i H(x) } e^{-(\beta_{i+1}-\beta_i) H(x) } = \frac{1}{Z(\beta_i)} \sum_{x \in \Omega} e^{-\beta_{i+1} H(x) }  = \frac{Z(\beta_{i+1})}{Z(\beta_i)} = \alpha_i. \]
The following result was shown by Dyer and Frieze~\cite{dyer92} (see~\cite{stefankovic09} for the statement here):

\begin{thm}
Let $Y_0,\dots,Y_{\ell-1}$ be independent random variables such that $\E[Y_i^2]/\E[Y_i]^2 \le B$ for all $i$, and write $\overline{Y} = \E[Y_0] \E[Y_1] \dots \E[Y_{\ell-1}]$. Let $\widetilde{\alpha}_i$ be the average of $16B\ell/\epsilon^2$ independent samples from $Y_i$, and set $\widetilde{Y} = \widetilde{\alpha}_0\,\widetilde{\alpha}_1 \dots \widetilde{\alpha}_{\ell-1}$. Then
\[ \Pr[(1-\epsilon)\overline{Y} \le \widetilde{Y} \le (1+\epsilon)\overline{Y}] \ge 3/4. \]
\end{thm}

Thus a classical algorithm can approximate $Z(\infty)$ up to relative error $\epsilon$ using $O(B \ell^2 / \epsilon^2)$ samples in total, assuming that $Z(0)$ can be computed without using any samples and that we have $\E[Y_i^2]/\E[Y_i]^2 \le B$. To characterise the latter  constraint, observe that we have
\[ \E[Y_i^2] = \frac{1}{Z(\beta_i)} \sum_{x \in \Omega} e^{-\beta_i H(x)} e^{-2(\beta_{i+1}-\beta_i) H(x)} = \frac{1}{Z(\beta_i)} \sum_{x \in \Omega} e^{(\beta_i-2\beta_{i+1}) H(x)} = \frac{Z(2\beta_{i+1}-\beta_i)}{Z(\beta_i)}, \]
so
\[ \frac{\E[Y_i^2]}{(\E[Y_i])^2} = \frac{Z(2\beta_{i+1}-\beta_i)Z(\beta_i)}{Z(\beta_{i+1})^2}. \]
This motivates the following definition:

\begin{dfn}[Chebyshev cooling schedules~\cite{stefankovic09}]
\label{dfn:chebysched}
Let $Z$ be a partition function. Let $\beta_0,\dots,\beta_\ell$ be a sequence of inverse temperatures such that $0 = \beta_0 < \beta_1 < \dots < \beta_\ell = \infty$. The sequence is called a $B$-Chebyshev cooling schedule for $Z$ if
\[ \frac{Z(2\beta_{i+1}-\beta_i)Z(\beta_i)}{Z(\beta_{i+1})^2} \le B \]
for all $i$, for some fixed $B$.
\end{dfn}

Assume that we have a sequence of estimates $\widetilde{\alpha}_i$ such that, for all $i$, $|\widetilde{\alpha}_i - \alpha_i| \le (\epsilon / 2\ell)\,\alpha_i$ with probability at least $1 - 1/(4\ell)$. We output as a final estimate
\[ \widetilde{Z} = Z(0)\,\widetilde{\alpha}_0\,\widetilde{\alpha}_1 \dots \widetilde{\alpha}_{\ell-1}. \]
By a union bound, all of the estimates $\widetilde{\alpha}_i$ are accurate to within $(\epsilon / 2\ell)\,\alpha_i$, except with probability at most $1/4$. Assuming that all the estimates are indeed accurate, we have
\[ 1-\epsilon/2 \le (1-\epsilon/(2\ell))^\ell \le \frac{\widetilde{Z}}{Z(\infty)} \le (1+\epsilon/(2\ell))^\ell \le e^{\epsilon/2} \le 1+\epsilon \]
for $\epsilon < 1$. Thus $|\widetilde{Z} - Z(\infty)| \le \epsilon\,Z(\infty)$ with probability at least $3/4$.

Using these ideas, we can formalise the discussion in Section \ref{sec:approxpart}.

\begin{thm}
Let $Z$ be a partition function with $|\Omega|=A$. Assume that we are given a $B$-Chebyshev cooling schedule $0 = \beta_0 < \beta_1 < \dots < \beta_\ell = \infty$ for $Z$. Further assume that we have the ability to exactly sample from the distributions $\pi_i$, $i=1,\dots,\ell-1$. Then there is a quantum algorithm which outputs an estimate $\widetilde{Z}$ such that
\[ \Pr[(1-\epsilon) Z(\infty) \le \widetilde{Z} \le (1+\epsilon)Z(\infty)] \ge 3/4. \]
using
\[ O\left(\frac{B\ell\log \ell}{\epsilon}\log^{3/2} \left(\frac{B\ell}{\epsilon} \right) \log \log \left(\frac{B\ell}{\epsilon} \right) \right)  = \widetilde{O}\left(\frac{B\ell^2}{\epsilon} \right) \]
samples in total.
\end{thm}

\begin{proof}
For each $i=1,\dots,\ell-1$, we use Algorithm \ref{alg:meanratio} to estimate $\E[Y_i]$ up to additive error $(\epsilon /(2\ell))\E[Y_i]$ with failure probability $1/(4\ell)$. As the $\beta_i$ form a $B$-Chebyshev cooling schedule, $\E[Y_i^2] / \E[Y_i]^2 \le B$, so $\Var(Y_i) / \E[Y_i]^2 \le B$. By Theorem \ref{thm:meanratio}, each use of Algorithm \ref{alg:meanratio} requires
\[ O\left(\frac{B\ell}{\epsilon}\log^{3/2} \left(\frac{B\ell}{\epsilon} \right) \log \log \left(\frac{B\ell}{\epsilon} \right) \log \ell \right) \]
samples from $\pi_i$ to achieve the desired accuracy and failure probability. The total number of samples is thus
\[ O\left(\frac{B\ell^2\log \ell}{\epsilon}\log^{3/2} \left(\frac{B\ell}{\epsilon} \right) \log \log \left(\frac{B\ell}{\epsilon} \right) \right) \]
as claimed.
\end{proof}

% ------------------------------------------------------------------------------

\subsection{Approximate sampling}
\label{sec:approxsamp}

It is unfortunately not always possible to exactly sample from the distributions $\pi_i$. However, one classical way of approximately sampling from each of these distributions is to use a (reversible, ergodic) Markov chain which has unique stationary distribution $\pi_i$. Assume the Markov chain has relaxation time $\tau$, where $\tau := 1/(1 - |\lambda_1|)$, and $\lambda_1$ is the second largest eigenvalue in absolute value. Then one can sample from a distribution $\widetilde{\pi}_i$ such that $\|\widetilde{\pi}_i-\pi_i\| \le \epsilon$ using $O(\tau \log(1/(\epsilon\pi_{\min,i})))$ steps of the chain, where $\pi_{\min,i} = \min_x |\pi_i(x)|$~\cite{levin09}. We would like to replace the classical Markov chain with a quantum walk, to obtain a faster mixing time. A construction due to Szegedy~\cite{szegedy04} defines a quantum walk corresponding to any ergodic Markov chain, such that the dependence on $\tau$ in the mixing time can be improved to $O(\sqrt{\tau})$~\cite{richter07}. Unfortunately, it is not known whether in general the dependence on $\pi_{\min,i}$ can be kept logarithmic~\cite{richter07,dunjko15}. Indeed, proving such a result is likely to be hard, as it would imply a polynomial-time quantum algorithm for graph isomorphism~\cite{aharonov07a}.

Nevertheless, it was shown by Wocjan and Abeyesinghe~\cite{wocjan08} (improving previous work on using quantum walks for classical annealing~\cite{somma08}) that one can achieve relatively efficient quantum sampling if one has access to a sequence of slowly varying Markov chains.

\begin{thm}[Wocjan and Abeyesinghe~\cite{wocjan08}]
\label{thm:warmstart}
Let $M_0,\dots,M_r$ be classical reversible Markov chains with stationary distributions $\pi_0,\dots,\pi_r$ such that each chain has relaxation time at most $\tau$. Assume that $|\ip{\pi_i}{\pi_{i+1}}|^2 \ge p$ for some $p>0$ and all $i \in \{0,\dots,r-1\}$, and that we can prepare the state $\ket{\pi_0}$. Then, for any $\epsilon > 0$, there is a quantum algorithm which produces a quantum state $\ket{\widetilde{\pi}_r}$ such that $\| \ket{\widetilde{\pi}_r} - \ket{\pi_r}\ket{0^a} \| \le \epsilon$, for some integer $a$. The algorithm uses
\[ O( r \sqrt{\tau} \log^2 (r/\epsilon) (1/p) \log(1/p)) \]
steps in total of the quantum walk operators $W_i$ corresponding to the chains $M_i$.
\end{thm}

In addition, one can approximately reflect about the states $\ket{\pi_i}$ more efficiently still, with a runtime that does not depend on $r$. This will be helpful because Algorithm \ref{alg:meanratio} uses significantly more reflections than it does copies of the starting state.

\begin{thm}[Wocjan and Abeyesinghe~\cite{wocjan08}, see~\cite{wocjan09} for version here]
\label{thm:reflect}
Let $M_0,\dots,M_r$ be classical reversible Markov chains with stationary distributions $\pi_0,\dots,\pi_r$ such that each chain has relaxation time at most $\tau$. For each $i$ there is an approximate reflection operator $\widetilde{R}_i$ such that
\[ \widetilde{R}_i \ket{\phi} \ket{0^b} = (2 \proj{\psi} - I)\ket{\phi} \ket{0^b} + \ket{\xi}, \]
where $\ket{\phi}$ is arbitrary, $b = O((\log \tau) (\log 1/\epsilon))$, and $\ket{\xi}$ is a vector with $\|\ket{\xi}\| \le \epsilon$. The algorithm uses $O(\sqrt{\tau} \log (1/\epsilon))$ steps of the quantum walk operator $W_i$ corresponding to the chain $M_i$.
\end{thm}

In our setting, we can easily create the quantum state $\ket{\pi_0}$, which is the uniform superposition over all configurations $x$. We now show that the overlaps $|\ip{\pi_i}{\pi_{i+1}}|^2$ are large for all $i$. We go via the chi-squared divergence
\[ \chi^2(\nu,\pi) := \sum_{x \in \Omega} \pi(x)\left( \frac{\nu(x)}{\pi(x)}-1 \right)^2 = \sum_{x \in \Omega} \frac{\nu(x)^2}{\pi(x)} - 1. \]
As noted in~\cite{stefankovic09}, one can calculate that
\be \label{eq:chicheb} \chi^2(\pi_{i+1},\pi_i) = \frac{Z(\beta_i)Z(2\beta_{i+1} - \beta_i)}{Z(\beta_{i+1})^2} - 1. \ee
Therefore, if the $\beta_i$ values form a Chebyshev cooling schedule, $\chi^2(\pi_{i+1},\pi_i) \le B-1$ for all $i$. For any distributions $\nu$, $\pi$, we also have
\[ \frac{1}{\sqrt{\chi^2(\nu,\pi)+1}} = \frac{1}{\sqrt{\sum_{x \in \Omega} \nu(x) \frac{\nu(x)}{\pi(x)}} } \le \sum_{x \in \Omega} \nu(x) \sqrt{\frac{\pi(x)}{\nu(x)}} = \ip{\nu}{\pi} \]
by applying Jensen's inequality to the function $x \mapsto 1/\sqrt{x}$. So, for all $i$, $|\ip{\pi_i}{\pi_{i+1}}|^2 \ge 1/B$. Note that in~\cite{stefankovic09} it was necessary to introduce the concept of a reversible Chebyshev cooling schedule to facilitate ``warm starts'' of the Markov chains used in the algorithm. That work uses the fact that one can efficiently sample from $\pi_{i+1}$, given access to samples from $\pi_i$, if $\chi^2(\pi_i,\pi_{i+1}) = O(1)$; this is the reverse of the condition (\ref{eq:chicheb}). Here we do not need to reverse the schedule as the precondition $|\ip{\pi_i}{\pi_{i+1}}|^2 \ge \Omega(1)$ required for Theorem \ref{thm:warmstart} is already symmetric.

We are now ready to formally state our result about approximating partition functions. We assume that $\epsilon$ is relatively small to simplify the bounds; this is not an essential restriction.

\begin{thm}
Let $Z$ be a partition function. Assume we have a $B$-Chebyshev cooling schedule $\beta_0 = 0 < \beta_1 < \beta_2 < \dots < \beta_\ell = \infty$ for $B=O(1)$. Assume that for every inverse temperature $\beta_i$ we have a reversible ergodic Markov chain $M_i$ with stationary distribution $\pi_i$ and relaxation time upper-bounded by $\tau$. Further assume that we can sample directly from $M_0$. Then, for any $\delta > 0$ and $\epsilon = O(1/\sqrt{\log \ell})$, there is a quantum algorithm which uses
\[ O((\ell^2 \sqrt{\tau} / \epsilon) \log^{5/2} (\ell/\epsilon) \log (\ell/\delta) \log \log(\ell/\epsilon) ) = \widetilde{O}(\ell^2 \sqrt{\tau} / \epsilon) \]
steps of the quantum walks corresponding to the $M_i$ chains and outputs $\widetilde{Z}$ such that
\[ \Pr[(1-\epsilon) Z(\infty) \le \widetilde{Z} \le (1+\epsilon)Z(\infty)] \ge 1-\delta. \]
\end{thm}

\begin{proof}
For each $i$, we use Algorithm \ref{alg:meanratio} to approximate $\alpha_i$ up to relative error $\epsilon / (2\ell)$, with failure probability $\gamma$, for some small constant $\gamma$. This would require $R$ reflections about the state $\ket{\pi_{\beta_i}}$, for some $R$ such that $R = O((\ell / \epsilon) \log^{3/2} (\ell/\epsilon) \log \log (\ell/\epsilon))$, and $O(\log(\ell/\epsilon) \log \log (\ell/\epsilon))$ copies of $\ket{\pi_{\beta_i}}$.
%  = O(\epsilon R / (\ell \sqrt{\log (\ell/\epsilon)}))

Instead of performing exact reflections and using exact copies of the states $\ket{\pi_i}$, we use approximate reflections and approximate copies of $\ket{\pi_i}$. By Theorem \ref{thm:reflect}, $O(\sqrt{\tau} \log (1/\epsilon_r))$ walk operations are sufficient to reflect about $\ket{\pi_i}$ up to an additive error term of order $\epsilon_r$. By Theorem \ref{thm:warmstart}, as we have a Chebyshev cooling schedule, a quantum state $\ket{\widetilde{\pi}_i}$ such that $\| \ket{\widetilde{\pi}_i} - \ket{\pi_i}\ket{0^b}\| \le \epsilon_s$ can be produced using $O(\ell \sqrt{\tau} \log^2 (\ell/\epsilon_s))$ steps of the quantum walks corresponding to the Markov chains $M_0,\dots,M_i$.

We choose $\epsilon_r = \gamma/R$, $\epsilon_s = \gamma$. Then the final state of Algorithm \ref{alg:meanratio} using approximate reflections and starting with the states $\ket{\widetilde{\pi}_i}$ rather than $\ket{\pi_i}$ can differ from the final state of an exact algorithm by at most $R \epsilon_r + \epsilon_s = 2\gamma$ in $\ell_2$ norm. This implies that the total variation distance between the output probability distributions of the exact and inexact algorithms is at most $2\gamma$, and hence by a union bound that the approximation is accurate up to relative error $\epsilon/(2\ell)$ except with probability $3\gamma$. For each $i$, we then take the median of $O(\log (\ell/\delta))$ estimates to achieve an estimate which is accurate up to relative error $\epsilon/(2\ell)$ except with probability at most $\delta/\ell$. By a union bound, all the estimates are accurate up to relative error $\epsilon/(2\ell)$ except with probability at most $\delta$, so their product is accurate to relative error $\epsilon$ except with probability at most $\delta$.

The total number of steps needed to produce all the copies of the states $\ket{\widetilde{\pi}_i}$ required is thus
\[ O(\ell \cdot \ell \sqrt{\tau} (\log^2 \ell) \cdot \log(\ell/\epsilon) \log \log (\ell/\epsilon) \cdot \log (\ell/\delta)) \]
and the total number of steps needed to perform the reflections is
\[ O(\ell \cdot \sqrt{\tau} (\log R) \cdot R \cdot \log (\ell/\delta) ).  \]
Adding the two, substituting the value of $R$, and using $\epsilon = O(1/\sqrt{\log \ell})$, we get an overall bound of
\[ O((\ell^2 \sqrt{\tau} / \epsilon) \log^{5/2} (\ell/\epsilon) \log (\ell/\delta) \log \log(\ell/\epsilon) ) = \widetilde{O}(\ell^2 \sqrt{\tau} / \epsilon) \]
as claimed.
\end{proof}

We remark that, in the above complexities, we have chosen to take the number of quantum walk steps used as our measure of complexity. This is to enable a straightforward comparison with the classical literature, which typically uses a random walk step as its elementary operation for the purposes of measuring complexity~\cite{stefankovic09}. To implement each quantum walk step efficiently and accurately, two possible approaches are to use efficient state preparation~\cite{chiang09} or recently developed approaches to efficient simulation of sparse Hamiltonians~\cite{berry15}.

Finally, we mention that one could also consider a more general setting for approximate sampling. Imagine that we would like to approximate the mean $\mu$ of some random variable chosen according to some distribution $\pi$, but only have access to samples from a distribution $\widetilde{\pi}$ that approximates $\pi$ (using some method which, for example, might not be a quantum walk). In this case, one can show that the estimation algorithm does not notice the difference between $\widetilde{\pi}$ and $\pi$ and hence allows efficient estimation of $\mu$. See Appendix~\ref{app:stability} for the details.

% ------------------------------------------------------------------------------

\subsection{Computing a Chebyshev cooling schedule}

We still need to show that, given a particular partition function, we can actually find a Chebyshev cooling schedule. For this we simply use a known classical result:

\begin{thm}[\v{S}tefankovi\v{c}, Vempala and Vigoda~\cite{stefankovic09}]
\label{thm:printschedule}
Let $Z$ be a partition function. Assume that for every inverse temperature $\beta$ we have a Markov chain $M_\beta$ with stationary distribution $\pi_\beta$ and relaxation time upper-bounded by $\tau$. Further assume that we can sample directly from $M_0$. Then, for any $\delta>0$ and any $B=O(1)$, we can produce a $B$-Chebyshev cooling schedule of length
\[ \ell = O(\sqrt{\log A} (\log n)(\log \log A)) \]
with probability at least $1-\delta$, using at most
\[ Q = O((\log A)((\log n) + \log \log A)^5 \tau \log(1/\delta)) \]
steps of the Markov chains.
\end{thm}

We remark that a subsequent algorithm~\cite{huber12} improves the polylogarithmic terms and the hidden constant factors in the complexity. However, this algorithm assumes that we can efficiently generate independent samples from distributions approximating $\pi_\beta$ for arbitrary $\beta$. The most efficient general algorithm known~\cite{stefankovic09} for approximately sampling from arbitrary distributions $\pi_\beta$ uses ``warm starts'' and hence does not produce independent samples.

Combining all the ingredients, we have the following result:

\begin{cor}
\label{cor:partitions}
Let $Z$ be a partition function and let $\epsilon > 0$ be a desired precision such that $\epsilon = O(1/\sqrt{\log \log A})$. Assume that for every inverse temperature $\beta$ we have a Markov chain $M_\beta$ with stationary distribution $\pi_\beta$ and relaxation time upper-bounded by $\tau$. Further assume that we can sample directly from $M_0$. Then, for any $\delta > 0$, there is a quantum algorithm which uses
\begin{multline*}
O(((\log A) (\log^2 n)(\log \log A)^2 \sqrt{\tau} / \epsilon) \log^{5/2} ((\log A)/\epsilon) \log ((\log A)/\delta) \log \log((\log A)/\epsilon)\\
 + (\log A)((\log n) + \log \log A)^5 \tau \log(1/\delta)) )\\
 = \widetilde{O}((\log A) \sqrt{\tau} ( 1 / \epsilon + \sqrt{\tau}))
 \end{multline*}
steps of the $M_\beta$ chains and their corresponding quantum walk operations, and outputs $\widetilde{Z}$ such that
\[ \Pr[(1-\epsilon) Z(\infty) \le \widetilde{Z} \le (1+\epsilon)Z(\infty)] \ge 1-\delta. \]
\end{cor}

The best comparable classical result known is $\widetilde{O}((\log A) \tau / \epsilon^2)$~\cite{stefankovic09}. We therefore see that we have achieved a near-quadratic reduction in the complexity with respect to both $\tau$ and $\epsilon$, assuming that $\epsilon \le 1/\sqrt{\tau}$. Otherwise, we still achieve a near-quadratic reduction with respect to $\epsilon$.

Note that, if we could find a quantum algorithm that outputs a Chebyshev cooling schedule using $\widetilde{O}((\log A) \sqrt{\tau})$ steps of the Markov chains, Corollary \ref{cor:partitions} would be improved to a complexity of $\widetilde{O}((\log A) \sqrt{\tau}/\epsilon)$. It is instructive to note why this does not seem to be immediate. The classical algorithm for this problem~\cite{stefankovic09} needs to approximately sample from Markov chains $M_\beta$ for arbitrary values of $\beta$. To do this, it starts by fixing a nonadaptive Chebyshev cooling schedule $0 < \beta'_1 < \beta'_2 < \dots < \beta'_\ell = \infty$ of length $\ell = \widetilde{O}(\log A)$. When the algorithm wants to sample from $M_\beta$ with $\beta'_i < \beta < \beta'_{i+1}$, the algorithm uses an approximate sample from $M_{\beta'_i}$ as a ``warm start''. To produce one sample corresponding to each $\beta'_i$ value requires $\widetilde{O}(\ell \tau)$ samples, because each $M_{\beta'_i}$ also provides a warm start for $M_{\beta'_{i+1}}$. But, in the quantum case, this does not work because, by no-cloning, the states $\ket{\pi_{\beta'_i}}$ cannot be reused in this way to provide warm starts for multiple runs of the algorithm.

% ------------------------------------------------------------------------------

\subsection{Some partition function problems}
\label{sec:pfproblems}

%Note that we have the following counterintuitive result.

%For any partition function problem where we have a Chebyshev schedule of length $\ell$, using Theorem \ref{thm:warmstart} we can produce quantum samples from the stationary distributions $\pi_{\beta_i}$ produced along the way. %And we can reflect about these states using Theorem \ref{} in time proportional to $\sqrt{\tau}$. But this even applies if the standard classical way to sample from $\pi_{\beta_i}$ does not involve warm starts but instead a more direct argument via a ``cold start'' random walk.

In this section we describe some representative applications of our results to problems in statistical physics and computer science.

{\bf The ferromagnetic Ising model.} This well-studied statistical physics model is defined in terms of a graph $G = (V,E)$ by the Hamiltonian
\[ H(z) = -\sum_{(u,v) \in E} z_u z_v, \]
where $|V| = n$ and $z \in \{\pm 1\}^n$. A standard method to approximate the partition function of the Ising model uses the Glauber dynamics. This is a simple Markov chain with state space $\{\pm1\}^n$, each of whose transitions involves only updating individual sites, and whose stationary distribution is the Gibbs distribution
\[ \pi_\beta(z) = \frac{1}{Z(\beta)} e^{-\beta H(z)}. \]
This Markov chain, which has been intensively studied for decades, is known to mix rapidly in certain regimes~\cite{martinelli99}. Here we mention just one representative recent result:

\begin{thm}[Mossel and Sly~\cite{mossel13}]
\label{thm:glauber}
For any integer $d>2$, and inverse temperature $\beta>0$ such that $(d-1) \tanh \beta < 1$, the mixing time of the Glauber dynamics on any graph of maximum degree $d$ is $O(n \log n)$.
\end{thm}

(More precise results than Theorem \ref{thm:glauber} are known for certain specific graphs such as lattices~\cite{martinelli94}.) As we have $A = 2^n$, in the regime where $(d-1)\tanh \beta < 1$ the quantum algorithm approximates $Z(\beta)$ to within $\epsilon$ relative error in $\widetilde{O}(n^{3/2} / \epsilon + n^2)$ steps. The fastest known classical algorithm with rigorously proven performance bounds~\cite{stefankovic09} uses time $\widetilde{O}(n^2 / \epsilon^2)$. We remark that an alternative approach of Jerrum and Sinclair~\cite{jerrum93}, which is based on analysing a different Markov chain, gives a polynomial-time classical algorithm which works for any temperature, but is substantially slower.

{\bf Counting colourings.} Here we are given as input a graph $G$ with $n$ vertices and maximum degree $d$. We seek to approximately count the number of valid $k$-colourings of $G$, where a colouring of the vertices is valid if all pairs of neighbouring vertices are assigned different colours, and $k=O(1)$. In physics, this problem corresponds to the partition function of the Potts model evaluated at zero temperature. It is known that the Glauber dynamics for the Potts model mixes rapidly in some cases~\cite{frieze07}. One particularly clean result of this form is work of Jerrum~\cite{jerrum95} showing that this Markov chain mixes in time $O(n \log n)$ if $k > 2d$. As here $A = k^n$, we obtain a quantum algorithm approximating the number of colourings of $G$ up to relative error $\epsilon$ in $\widetilde{O}(n^{3/2}/ \epsilon + n^2)$ steps, as compared with the classical $\widetilde{O}(n^2 / \epsilon^2)$~\cite{stefankovic09}. 

{\bf Counting matchings.} A matching in a graph $G$ is a subset $M$ of the edges of $G$ such that no pair of edges in $M$ shares a vertex. In statistical physics, matchings are often known as monomer-dimer coverings~\cite{heilmann72}. To count the number of matchings, we consider the partition function
\[ Z(\beta) = \sum_{M \in \mathcal{M}} e^{-\beta |M|}, \]
where $\mathcal{M}$ is the set of matchings of $G$. We have $Z(0) = |\mathcal{M}|$, while $Z(\infty) = 1$, as in this case  the sum is zero everywhere except the empty matching ($0^0=1$). Therefore, in this case we seek to approximate $Z(0)$ using a telescoping product which starts with $Z(\infty)$. In terms of the cooling schedule $0 = \beta_0 < \beta_1 < \dots < \beta_\ell = \infty$, we have
\[ Z(\beta_0) = Z(\beta_\ell) \frac{Z(\beta_{\ell-1})}{Z(\beta_{\ell})} \frac{Z(\beta_{\ell-2})}{Z(\beta_{\ell-1})} \dots \frac{Z(\beta_0)}{Z(\beta_1)}. \]
As we have reversed our usage of the cooling schedule, rather than looking for it to be a $B$-Chebyshev cooling schedule we instead seek the bound
\[ \frac{Z(2\beta_i - \beta_{i+1})Z(\beta_{i+1})}{Z(\beta_i)^2} \le B \]
to hold for all $i = 0, \dots, \ell-1$. That is, the roles of $\beta_i$ and $\beta_{i+1}$ have been reversed as compared with (\ref{dfn:chebysched}). However, the classical algorithm for printing a cooling schedule can be modified to output a ``reversible'' schedule where this constraint is satisfied too, with only a logarithmic increase in complexity~\cite{stefankovic09}. In addition, it was shown by Jerrum and Sinclair~\cite{jerrum89,jerrum03} that, for any $\beta$, there is a simple Markov chain which has stationary distribution $\pi$, where
\[ \pi(M) = \frac{1}{Z(\beta)} \sum_{M \in \mathcal{M}} e^{-\beta |M|}, \]
and which has relaxation time $\tau = O(nm)$ on a graph with $n$ vertices and $m$ edges. Finally, in the setting of matchings, $A = O(n! 2^n)$. Putting these parameters together, we obtain a quantum complexity $\widetilde{O}(n^{3/2} m^{1/2}/\epsilon + n^2m)$, as compared with the lowest known classical bound $\widetilde{O}(n^2 m / \epsilon^2)$~\cite{stefankovic09}.

% ------------------------------------------------------------------------------

\section{Estimating the total variation distance}
\label{sec:tvd}

Here we give the technical details of our improvement of the accuracy of a quantum algorithm of Bravyi, Harrow and Hassidim~\cite{bravyi11a} for estimating the total variation distance between probability distributions. In this setting, we are given the ability to sample from probability distributions $p$ and $q$ on $n$ elements, and would like to estimate $\|p-q\| := \frac{1}{2} \|p-q\|_1 = \frac{1}{2} \sum_{x \in [n]} |p(x) - q(x)|$ up to additive error $\epsilon$. Classically, estimating $\|p-q\|$ up to error, say, $0.01$ cannot be achieved using $O(n^\alpha)$ samples for any $\alpha<1$~\cite{valiant11}, but in the quantum setting the dependence on $n$ can be improved quadratically:

\begin{thm}[Bravyi, Harrow and Hassidim~\cite{bravyi11a}]
Given the ability to sample from $p$ and $q$, there is a quantum algorithm which estimates $\|p-q\|$ up to additive error $\epsilon$, with probability of success $1-\delta$, using $O(\sqrt{n}/(\epsilon^8 \delta^5))$ samples.
\end{thm}

Here we will use Theorem \ref{thm:meanerr} to improve the dependence on $\epsilon$ and $\delta$ of this algorithm. We will approximate the mean output value of the following algorithm, which was a subroutine previously used in~\cite{bravyi11a}.

\boxalgm{alg:l1dist}{Subroutine for estimating the total variation distance}{
Let $p$ and $q$ be probability distributions on $n$ elements and let $r = (p+q)/2$.
\begin{enumerate}
\item Draw a sample $x \in [n]$ according to $r$.
\item Use amplitude estimation with $t$ queries, for some $t$ to be determined, to obtain estimates $\widetilde{p}(x)$, $\widetilde{q}(x)$ of the probability of obtaining outcome $x$ under distributions $p$ and $q$.
\item Output $|\widetilde{p}(x) - \widetilde{q}(x)| / (\widetilde{p}(x) + \widetilde{q}(x))$.
\end{enumerate}
}

If the estimates $\widetilde{p}(x)$, $\widetilde{q}(x)$ were precisely accurate, the expected output of the subroutine would be
\[ E := \sum_{x \in [n]} \left( \frac{p(x) + q(x)}{2}\right) \frac{ | p(x) - q(x) |}{p(x) + q(x)} = \frac{1}{2} \sum_{x \in [n]} | p(x) - q(x) | = \|p-q\|. \]
We now bound how far the expected output $\widetilde{E}$ of the algorithm is from this exact value. By linearity of expectation,
\[ | \widetilde{E} - E | = \left| \sum_{x \in [n]} r(x) \E[ \widetilde{d}(x) - d(x)]\right| \le \sum_{x \in [n]} r(x) \E[|\widetilde{d}(x) - d(x)|] \]
where $d(x) = |p(x) - q(x)|/(p(x)+q(x))$, $\widetilde{d}(x) = |\widetilde{p}(x) - \widetilde{q}(x)|/(\widetilde{p}(x)+\widetilde{q}(x))$. Note that $\widetilde{d}(x)$ is a random variable. Split $[n]$ into ``small'' and ``large'' parts according to whether $r(x) \le \epsilon / n$. Then 
\beas
| \widetilde{E} - E | &\le& \sum_{x, r(x)\le \epsilon / n} r(x) \E [|\widetilde{d}(x) - d(x)|] + \sum_{x, r(x)\ge \epsilon / n} r(x) \E[|\widetilde{d}(x) - d(x)|] \\
&\le& \epsilon + \sum_{x, r(x)\ge \epsilon / n} r(x) \E[|\widetilde{d}(x) - d(x)|]
\eeas
using that $0 \le d(x), \widetilde{d}(x) \le 1$. From Theorem \ref{thm:ampest}, for any $\delta > 0$ we have
\[ |\widetilde{p}(x) - p(x)| \le 2 \pi \frac{\sqrt{p(x)}}{t} + \frac{\pi^2}{t^2} \]
except with probability at most $\delta$, using $O(t \log 1/\delta)$ samples from $p$. If $t \ge 4\pi / (\eta \sqrt{p(x)+q(x)})$ for some $0 \le \eta \le 1$, this implies that
\[ |\widetilde{p}(x) - p(x)| \le \frac{2\pi \eta \sqrt{p(x)} \sqrt{p(x) + q(x)}}{4 \pi} + \frac{\pi^2\eta^2 (p(x) + q(x))}{16\pi^2} \le \eta(p(x)+q(x)) \]
except with probability at most $\delta$. A similar claim also holds for $|\widetilde{q}(x) - q(x)|$. We now use the following technical result from~\cite{bravyi11a}:

\begin{prop}
\label{prop:ratioest}
Consider a real-valued function $f(p,q) = (p-q)/(p+q)$ where $0 \le p,q \le 1$. Assume that $|p - \widetilde{p}|, |q-\widetilde{q}| \le \eta(p+q)$ for some $\eta \le 1/5$. Then $|f(p,q) - f(\widetilde{p},\widetilde{q})| \le 5\eta$.
\end{prop}

By Proposition \ref{prop:ratioest}, for all $x$ such that $t \ge 4\pi / (\eta \sqrt{p(x)+q(x)})$ we have $|\widetilde{d}(x)-d(x)| \le 5\eta$, except with probability at most $2\delta$. We now fix $t = \lceil 10 \sqrt{2} \pi \sqrt{n} \epsilon^{-3/2} \rceil$. Then, for all $x$ such that $p(x) + q(x) \ge 2\epsilon/n$, $|\widetilde{d}(x)-d(x)| \le \epsilon$ except with probability at most $2\delta$. Thus, for all $x$ such that $r(x) \ge \epsilon/n$,
\[ \E[|\widetilde{d}(x)-d(x)|] \le 2\delta + (1-2\delta)\epsilon \le 2\delta + \epsilon. \]
Taking $\delta = \epsilon$, we have
\[ | \widetilde{E} - E | \le 4 \epsilon \]
for any $\epsilon$, using $O(\sqrt{n} \epsilon^{-3/2} \log(1/\epsilon))$ samples. It therefore suffices to use $O(\sqrt{n} \epsilon^{-3/2} \log(1/\epsilon))$ samples to achieve $| \widetilde{E} - E | \le \epsilon/2$. As the output of this subroutine is bounded between 0 and 1, to approximate $\widetilde{E}$ up to additive error $\epsilon/2$ with failure probability $\delta$, it suffices to use the subroutine $O((1/\epsilon) \log (1/\delta))$ times by Theorem \ref{thm:meanerr}. So the overall complexity is $O((\sqrt{n} / \epsilon^{5/2}) \log(1/\epsilon) \log(1/\delta))$. For small $\epsilon$ and $\delta$ this is a substantial improvement on the $O(\sqrt{n} / (\epsilon^8 \delta^5))$ complexity stated by Bravyi, Harrow and Hassidim~\cite{bravyi11a}.

% ------------------------------------------------------------------------------

\subsection*{Acknowledgements}

This work was supported by the UK EPSRC under Early Career Fellowship EP/L021005/1. I would like to thank Aram Harrow for helpful conversations and pointing out references, and Daniel Lidar for supplying further references. I would also like to thank several anonymous referees for their helpful comments. Special thanks to Tongyang Li for pointing out an error in Section \ref{sec:tvd}.

% ------------------------------------------------------------------------------

\appendix

\section{Stability of Algorithm \ref{alg:meanvariance}}
\label{app:stability}

It is often the case that one wishes to estimate some quantity of interest defined in terms of samples from some probability distribution $\pi$, but can only sample from a distribution $\widetilde{\pi}$ which is close to $\pi$ in total variation distance (for example, using Markov chain Monte Carlo methods). We now show that, if Algorithm \ref{alg:meanvariance} is given access to samples from $\widetilde{\pi}$ rather than $\pi$, it does not notice the difference. We will need the following claim.

\begin{claim}
\label{claim:arcsin}
For any $x,y \in [0,1]$,
\[ |\arcsin x - \arcsin y | \le \frac{\pi}{2} \sqrt{|x^2-y^2|}. \]
\end{claim}

\begin{proof}
We use a standard addition formula for arcsin to obtain
\beas
|\arcsin x - \arcsin y | &=& | \arcsin(x \sqrt{1-y^2} - y\sqrt{1-x^2})|\\
&\le& \frac{\pi}{2} |\sqrt{x^2(1-y^2)} - \sqrt{y^2(1-x^2)}|\\
&\le& \frac{\pi}{2} \sqrt{|x^2 - y^2|},
\eeas
where the first inequality is $\sin \theta \ge (2/\pi) \theta$ for all $\theta \in [0,\pi/2]$, and the second inequality is
\[ |a-b| \le \sqrt{|a-b|(a+b)} = \sqrt{|a^2-b^2|}, \]
valid for all non-negative $a$ and $b$.
\end{proof}

\begin{lem}
\label{lem:indist}
Let $\mathcal{A}$ and $\mathcal{B}$ be algorithms with distributions $\mathcal{D}_{\mathcal{A}}$ and $\mathcal{D}_{\mathcal{B}}$ on their output values, such that $\| \mathcal{D}_{\mathcal{A}} - \mathcal{D}_{\mathcal{B}}\| \le \gamma$, for some $\gamma$. Assume that Algorithm \ref{alg:meanvariance} is applied to $\mathcal{A}$, and uses the operator $U = 2 \proj{\psi}-I$ $T$ times, where $\ket{\psi} = \mathcal{A}\ket{0}$. Then the algorithm estimates $\E[v(\mathcal{B})]$ up to additive error $\epsilon$ except with probability at most $3/10 + \frac{\pi^2}{\sqrt{6}} T \sqrt{\gamma}$.
\end{lem}

Lemma \ref{lem:indist} is reminiscent of the hybrid argument for proving lower bounds on quantum query complexity~\cite{bennett97}: if the distributions $\mathcal{D}_{\mathcal{A}}$ and $\mathcal{D}_{\mathcal{B}}$ are close, and the amplitude amplification algorithm makes few queries, it cannot distinguish them. However, here the quantifiers appear in a different order: whereas in the setting of lower-bounding quantum query complexity we wish to show that there exist pairs of distributions which are indistinguishable by any possible algorithm, here we wish to show that one fixed algorithm cannot distinguish any pair of close distributions.

Also note that Wocjan et al.~\cite{wocjan09} proved a similar result in the setting where we are given access to an approximate rotation $\widetilde{U} \approx U$. However, the result here is more general, in that we do not assume that $\ket{\phi} = \mathcal{B}\ket{0}$ is close to $\ket{\psi}$, but merely that the measured probability distributions are close.

\begin{proof}
We first use the calculations for the output probabilities of the amplitude estimation algorithm from~\cite{brassard02} when applied as in Theorem \ref{thm:meanerr} with $t$ queries to an algorithm with mean output value $\mu_A$, and another with mean output value $\mu_B$.

For $x,y \in \R$, define $d(x,y) = \min_{z \in \Z} |z + x - y |$. $2 \pi d(x,y)$ is the length of the shortest arc on the unit circle between $e^{2 \pi i x}$ and $e^{2 \pi i y}$. Let $\omega_A$ and $\omega_B$ be defined by $\sin^2 \omega_A = \mu_A$, $\sin^2 \omega_B = \mu_B$, and set $\Delta = d(\omega_A,\omega_B)$. Finally, let $\mathcal{M}_\mathcal{A}$ and $\mathcal{M}_\mathcal{B}$ be the distributions over the measurement outcomes when amplitude estimation is applied to estimate $\mu_A$, $\mu_B$.

The distribution on the measurement outcomes of the amplitude estimation algorithm after $t$ uses of the input operator, when applied to a phase of $\omega$, is equivalent~\cite{brassard02} to that obtained by measuring the state
\[ \ket{\mathcal{S}_t(\omega)} := \frac{1}{\sqrt{t}} \sum_{y \in [t]} e^{2\pi i \omega y} \ket{y}, \]
so the total variation distance between the distributions $\mathcal{M}_\mathcal{A}$ and $\mathcal{M}_\mathcal{B}$ obeys the bound
\[ \|\mathcal{M}_\mathcal{A} - \mathcal{M}_\mathcal{B}\|^2 \le 1 - |\ip{\mathcal{S}_t(\omega_A)}{\mathcal{S}_t(\omega_B)}|^2 = 1 - \frac{\sin^2(\pi t \Delta)}{t^2 \sin^2(\pi \Delta)}, \]
where the first equality is standard~\cite{nielsen00} and the second equality is~\cite[Lemma 10]{brassard02}. Using the inequalities
\[ \theta - \frac{\theta^3}{6} \le \sin \theta \le \theta, \]
valid for $\theta \ge 0$, we obtain
\[ \|\mathcal{M}_\mathcal{A} - \mathcal{M}_\mathcal{B}\|^2 \le 1 - \left( \frac{\pi t \Delta - (\pi t \Delta)^3 / 6}{t \pi \Delta}\right)^2 = 1 - \left(1 - \frac{(\pi t \Delta)^2}{6} \right)^2 \le \frac{(\pi t \Delta)^2}{3}. \]
As we have
\[ \Delta  = \min_{z \in \Z} |z + \omega_A - \omega_B | \le |\omega_A - \omega_B| \le \frac{\pi}{2} \sqrt{|\mu_A - \mu_B|} \]
by Claim \ref{claim:arcsin}, we have
\[  \|\mathcal{M}_\mathcal{A} - \mathcal{M}_\mathcal{B}\| \le \frac{\pi^2}{2 \sqrt{3}} t \sqrt{|\mu_A - \mu_B|}. \]
Within Algorithm \ref{alg:meansub}, Theorem \ref{thm:meanerr} is applied to $v(\mathcal{A}_{2^{\ell-1},2^{\ell}})/2^\ell$ for various values of $\ell$. We have
\beas
| \E[v(\mathcal{A}_{2^{\ell-1},2^{\ell}})/2^\ell] - \E[v(\mathcal{B}_{2^{\ell-1},2^{\ell}})/2^\ell] | &=& \frac{1}{2^\ell} \sum_{2^{\ell-1} \le x < 2^\ell} x | \Pr[v(\mathcal{A}) = x] - \Pr[v(\mathcal{B}) = x]|\\
&\le& \sum_x | \Pr[v(\mathcal{A}) = x] - \Pr[v(\mathcal{B}) = x]|\\
&=& 2 \| \mathcal{D}_\mathcal{A} - \mathcal{D}_\mathcal{B}\| \le 2 \gamma.
\eeas
Thus, for each run of the algorithm which uses $\mathcal{A}$ $t$ times,
\[  \|\mathcal{M}_\mathcal{A} - \mathcal{M}_\mathcal{B}\| \le \frac{\pi^2}{\sqrt{3}} t \sqrt{\gamma}. \]
This is equivalent to the output of the algorithm being a probabilistic mixture of $\mathcal{M}_\mathcal{B}$ and some other distribution $\mathcal{M}$, where the probability of it being $\mathcal{M}$ is at most $\frac{\pi^2}{\sqrt{3}} t \sqrt{\gamma}$.

Algorithm \ref{alg:meanvariance} uses $\mathcal{A}$ $T$ times in total. Each use of $\mathcal{A}$ is either within Algorithm \ref{alg:meansub} or one separate sample from $v(\mathcal{A})$ in Algorithm \ref{alg:meanvariance}. We can similarly think of this sample as being taken from $\mathcal{B}$, except with probability at most $\gamma \le \frac{\pi^2}{\sqrt{3}} \sqrt{\gamma}$. Taking a union bound over all uses of $\mathcal{A}$, we get the claimed bound.
\end{proof}

% ------------------------------------------------------------------------------
% ------------------------------------------------------------------------------

%\putbib[../../thesis.bib]
%\end{bibunit}

\bibliographystyle{plain}
\bibliography{../../thesis}

\end{document}